\documentclass[letterpaper,11pt]{article}
\usepackage{subfigure}
\usepackage{float}
\usepackage{graphicx}
\usepackage{tikz}
\usepackage{enumerate}
\usepackage{bbm}
\usepackage{amsmath}
\usepackage{amsthm}
\usepackage{amssymb}
\usepackage{relsize}
\usepackage{multirow}

\newcommand\F{\ensuremath{\mathcal{F}}}

\newcommand\E{\ensuremath{\mathbb{E}}}
\newcommand\R{\ensuremath{\mathbb{R}}}
\newcommand\fT{\ensuremath{\mathfrak{T}}}

\newcommand\cI{\ensuremath{\mathcal{I}}}
\newcommand\cO{\ensuremath{\mathcal{O}}}

\newcommand\eps{\epsilon}
\newcommand\talph{\tilde{\alpha}}
\newcommand\sqc{\sqrt{c}}

\newcommand\PP{\ensuremath{\mathbb{P}}}

\DeclareMathOperator\sgn{sign}

\newtheorem{lemma}{Lemma}
\newtheorem{corollary}{Corollary}
\newtheorem{proposition}{Proposition}

\theoremstyle{remark}
\newtheorem{remark}{Remark}

\begin{document}

\title{Optimal Execution with Dynamic Order Flow Imbalance}
\author{Kyle Bechler \\ Statistics and Applied Probability\\
  University of California, Santa Barbara\\ \and Michael Ludkovski\\
  Statistics and Applied Probability\\
  University of California, Santa Barbara\\
}
\date{\today}

\maketitle

\begin{abstract}
We examine optimal execution models that take into account both market microstructure impact and informational costs.
Informational footprint is related to order flow and is represented by the trader's influence on the flow imbalance process, while microstructure influence is captured by instantaneous price impact. We propose a continuous-time stochastic control problem that balances between these two costs. Incorporating order flow imbalance leads to the consideration of the current market state and specifically whether one's orders lean with or against the prevailing order flow, key components often ignored by execution models in the literature. In particular, to react to changing order flow, we endogenize the trading horizon $T$. After developing the general indefinite-horizon formulation, we investigate several tractable approximations that sequentially optimize over price impact and over $T$. These approximations, especially a dynamic version based on receding horizon control, are shown to be very accurate and connect to the prevailing Almgren-Chriss framework. We also discuss features of empirical order flow and links between our model and ``Optimal Execution Horizon'' by Easley et al (Mathematical Finance, 2013).


\end{abstract}

\section{Introduction}

The concept of optimal execution in financial markets is concerned with realizing the best value for the asset traded via optimization with respect to the incurred trading costs. These costs are driven by two fundamental components:
market microstructure frictions and informational asymmetries. Market microstructure implies that market liquidity is finite and trades generate price impact. Informational costs reflect the fact that trades are observed by other participants who will then adjust their own strategies and views of the asset value and create adverse selection.

The complexity of today's market micro-structure has led to a wide variety of approaches that focus on different facets of the problem. 
The seminal work of Almgren and Chriss \cite{almgren2012optimal,almgren2001optimal,almgren2003optimal} originated the framework of diffusive mid-price models where trading is continuous and price impact depends on the trading {rate}. Extensions of this model to account for limit order book resiliency can be found in \cite{alfonsi2010optimal}, and for other transient price impacts in \cite{gatheral2012transient}. Alternatively, focusing on trading via limit orders, \cite{bayraktar2011optimal,cartea2011buy,gueant2012optimal} have investigated execution under fill risk with discontinuous inventory trajectories. Hybrid models that mix limit and market orders appear in \cite{cartea2014optimal,guilbaud2013optimal,Huitema13}.

The above approaches provide frameworks for modelling price impact and execution costs but say nothing about informational costs. While there are some studies \cite{moallemi2013cost} that analyze adverse selection, so far research in this direction has been mostly confined to empirical studies of long-term (``permanent'') price impact. For example, Bouchaud et al.~\cite{bouchaud2008markets,bouchaud2010price} and Farmer et al.~\cite{Farmer} analyze the characteristics of order flow data, asserting that order flow is predictable over time-scales of up to several days and that order flow plays an important role in the formation of prices. A notable recent contribution is by Easley at al.~\cite{easley2012optimal}, who minimize a combination of information leakage and volatility risk in a static setting. Importantly, \cite{easley2012optimal} assert that the liquidity process is sensitive to the execution strategy of the trader with the linkage through order flow.   Specifically, the theory of market toxicity
 \cite{easley2012flow} implies that the expected absolute imbalance in order flow can be used to approximate the probability that a trade is information based. Therefore one-sided markets are likely to create adverse selection and consequently in such toxic conditions
market makers widen the range at which they provide liquidity. A qualitatively similar result was shown by Cartea et al.~\cite{cartea2011buy} where the market maker's order placement shifts deeper into the limit order book in response to one-sided market order flow.  It follows that by influencing the order flow, the trader's actions create a transient impact on market toxicity and hence execution costs.

In the present paper we bridge the above ideas within a combined dynamic model. We maintain the basic features of the Almgren-Chriss setup, including continuous trading and (quadratic) instantaneous price impact arising from market microstructure. However, we also incorporate a novel stochastic factor $(Y_t)$ for order flow, which is similarly impacted from executed trades. The  transient impact on $Y$ represents the informational footprint of the trader and introduces a feedback loop into the problem. It allows the trader to react to changing market conditions, in particular markets changing from being buy-driven to ask-driven and vice versa. This order flow feedback is anecdotally incorporated in most implemented practitioner algorithms and has been questioned in connection with liquidity crises (notably the Flash Crash, see~\cite{kirilenko2011flash}). We use our model to examine the dynamic problem of \emph{Optimal Execution Horizon} that was introduced in a 1-period version by Easley et al.~\cite{easley2012optimal}. Thus, in contrast to Almgren-Chriss and typical execution models, we endogenize the execution horizon which is optimally chosen depending on market liquidity. Intuitively, trading should slow down when informational costs are high, and speed up when they are low. From that point of view, rather than a pure optimization problem, optimal execution is about trading-off price impact and information leakage against timing risk. Moreover, endogenizing execution horizon generates dynamic execution strategies even if the underlying asset value is a martingale. This allows to obtain inherently adaptive trading in contrast to deterministic strategies in popular models (e.g.~\cite{alfonsi2010optimal,almgren2001optimal,almgren2003optimal}) that consider only price impact.

The rest of the paper is organized as follows. In Section \ref{sec:model} we setup our dynamic order flow model and derive the associated HJB equation with endogenous horizon. In Section \ref{sec:linear-quadratic} we address the respective fixed horizon problem which allows us to obtain several closed-form approximate strategies once we optimize over the execution horizon. The latter is done in Section \ref{sec:optimal-T} where we also compare several strategies that result.  Section \ref{sec:implement} looks at some stylized facts about order flow and discusses calibration.  

\section{The Optimal Execution Problem}\label{sec:model}
The problem in view is liquidation of a position of size $x=x_0$. We assume a continuous-time setup, with trading taking place continuously and via infinitesimal amounts. Namely, the trader trades $\dot{x}_t dt$ shares at time $t$, so that his inventory $x_t$ follows the dynamics
\begin{align}
dx_t = \dot{x}_t \,dt.
\end{align}
Execution ends at the random horizon $$T_0 := \inf \{ t \ge 0 : x_t = 0\},$$ whereupon inventory is exhausted. Throughout, time is supposed to be in traded-volume units.


Beyond the inventory $x_t$, the main state variable of our model is the order flow imbalance $Y_t$. The order flow imbalance captures the intrinsic fluctuations among supply and demand for the security realized by the unequal amounts of buyer- and seller-initiated orders and trades. On the short time-scale (intra-day to several days) it is empirically quasi-stationary, in the sense that the observed volume is several orders of magnitude larger than the deviations in net imbalance, cf.~Section \ref{sec:empirical-flow}. Moreover it is highly noisy and appears to be mean-reverting to zero. Therefore, we choose to model $(Y_t)$ in terms of a mean-zero stationary process.

Let $Y^0$ represent the flow imbalance in the absence of the trader. As a starting point we take $(Y^0_t)$ to be an Ornstein-Uhlenbeck process with mean-reversion parameter $\beta$,
\begin{align}\label{eq:y0}
dY^0_t = - \beta Y^0_t \, dt + \sigma \, dW_t.
\end{align}
The mean-reversion strength $\beta$ controls the time-scale of the memory in flow imbalance, while the volatility $\sigma$ controls the size of fluctuations in flow imbalance. Since imbalance is intuitively in the range $[-1,1]$ (representing markets with 100\% buyers, and 100\% sellers respectively), the fraction $\sigma^2/(2\beta)$, which is the stationary variance of $Y^0$, should be on the order of $\sigma^2/(2\beta) \in [0.01,0.2]$.

%
%
The execution program of the trader introduces a downward pressure on the order imbalance process as a result of his selling. The information leaked by the trader's action creates a drift in the realized order imbalance $Y_t$, pushing it below $Y^0_t$. By displacing other orders, the trader impacts expectations regarding future order flows and generates adverse selection. A more precise description of this mechanism using (more realistic) discrete-time setup and discrete trades is presented in Section \ref{sec:discrete-time}. We postulate that
\begin{equation}\label{dyn}
dY_t = -\beta Y_t dt - \phi(\dot{x}_t) dt + \sigma dW_t,
\end{equation}
where $\phi(\cdot)$ captures the information leakage. Observe that
\begin{align}
Y_t = Y_t^0 + \int_0^t e^{\beta(s-t)} \phi(\dot{x}_s) \, ds,
\end{align}
so the execution program generates an exponentially decaying impact on $Y^0$. We assume that $\phi: \R_+ \to \R_+$ is non-decreasing with $\phi(0) = 0$. The three main cases we consider are: $\phi(\dot{x}_t) \equiv 0$ corresponding to zero informational footprint; $\phi(\dot{x}_t) = \phi_t$ corresponding to deterministic (but non-zero) impact; and proportional (linear) impact $\phi(\dot{x}_t)= \eta \dot{x}_t$.  Linear impact is computationally convenient, though not necessarily the most realistic reflection of how a trader's activity might influence expectations regarding order flow.  Impact that depends on the current value of the flow imbalance is investigated in Section \ref{sec:discrete-time}.

The information leakage is ``abstract'' in the sense that it does not generate trading costs per se. However, in line with \cite{easley2012optimal} we assume that there are adverse selection costs associated with trading in an unbalanced market. Here we assume that this cost is symmetric in $Y_t$ (but note that agent's actions induce only one-sided effects of $Y_t$); for tractability we take it quadratic. In addition, we carry through the two costs from Almgren-Chriss: instantaneous impact $g(\dot{x}_t)$ of trading at rate $\dot{x}_t$, and inventory cost $\lambda(x_t)$ for carrying a position of $x_t$ at $t$. 

The continuous-time execution problem is to minimize the sum of the corresponding expected execution costs
\begin{equation} \label{oppr}
 \inf_{(x_t) \in \mathcal{X}(x)} \E_{x,y} \left[\int_{0}^{T_0} \left(g(\dot{x}_s)+ \kappa Y^2_s+\lambda(x_s)\right) \, ds\right],
\end{equation}
over admissible execution strategies $(x_t) \in \mathcal{X}(x)$. The above expectation is conditional on an initial value $Y_0=y$ and initial inventory $x$ which induce the measure $\mathbb{P}_{x,y}$. The horizon $T_0$ is part of the solution, so that the optimization is formally taking place on the whole future $s \in [0,\infty)$. While the first term in \eqref{oppr} incentivizes the trader to slow down in order to reduce his immediate liquidity costs, the next two terms of the cost functional are such that under certain market conditions, it may be optimal to accelerate trading in order to exit the market sooner.

Let $\F_t = \sigma(Y_s: s \le t)$ denote the filtration generated by $Y$. Admissible strategies $(x_t) \in \mathcal{X}(x)$, consist of $(\F_t)$-progressively measurable, absolutely continuous trajectories $t\mapsto x_t$, such that $x_0=x$, $ \lim_{t \to \infty} x_t=0$ and $\int_{0}^{\infty}g(\dot{x}_s) \, ds< \infty$ $\PP$-a.s. We also require the following assumptions on the model ingredients:
\begin{itemize}
\item Instantaneous impact function $g:\mathbb{R}_+ \mapsto \mathbb{R}_+ $ is strictly convex;
\item The informational cost parameter $\kappa \ge 0$;
\item Inventory risk $\lambda:\mathbb{R}_+ \mapsto \mathbb{R}_+$ is non-decreasing in $x$.
\end{itemize}

The cost functional in \eqref{oppr} is consistent with other approaches taken in the literature.  The assumption that $g(\dot{x})$ is convex matches the empirical fact that market participants like to divide a large ``parent'' order into smaller orders in order to reduce trading costs. In limit order books (LOB), $g(\cdot)$ represents the depth of the LOB on the ask-side. If this depth is constant, the instantaneous trading cost is quadratic $g(\dot{x}) = \dot{x}^2$ (by rescaling $\kappa$ and $\lambda$ we assume without loss of generality that the leading coefficient is 1). This assumption also appears in \cite{almgren2001optimal,cartea2014optimal,gatheral2011optimal} among others. Because our primary focus is on informational costs, we assume for the moment that there are no other transient/permanent impacts on the asset value $S_t$ and further posit that strategies are independent of asset dynamics. In Section \ref{sec:conclusion} we return to this issue and discuss extensions that allow for positive correlation between asset price dynamics $(S_t)$ and order flow $(Y_t)$.


Our second cost term $\kappa Y_t^2$ is motivated primarily by the model presented in \cite{easley2012optimal} and captures the cost of information leakage. The main premise is that liquidity costs (e.g.~likelihood of adverse selection) are higher in markets with unbalanced flow. However, unlike \cite{easley2012optimal}, our model incorporates mean reversion and stochasticity in the flow imbalance process $(Y_t)$.  So the impact being modelled by this term is transient, and somewhat related to models that investigate resiliency in the limit order book \cite{alfonsi2010optimal,gatheral2012transient}. For other papers that considered stochastic execution costs see \cite{almgren2012optimal}.

   The last term $\lambda(x_t)$ in \eqref{oppr} represents timing risk, penalizing the trader for leaving his position exposed to adverse price movements.   Several risk terms have been applied within the execution literature.  The seminal work by Almgren and Chriss, which optimizes over a mean-variance cost functional, reduces to a calculus of variations problem and the risk term $\lambda(x ) = cx^2 $.  Gatheral and Schied~\cite{gatheral2011optimal} investigated a time-weighted value-at-risk measure proportional to $\lambda(x) = c x$. In the context of timing risk, $\lambda(x) = c$ generates costs that are proportional to execution time which is a non-trivial modification once the horizon $T_0$ is not fixed.

\begin{remark}
In  \eqref{oppr} we restrict our attention to pure selling strategies, so that $x_t$ is non-increasing. As such all our cost functionals are only defined for positive selling rates. Depending on the strength of information leakage, it is possible that the constraint $\dot{x}_t \ge 0$ is binding, i.e.~execution is suspended until market conditions improve. It is beyond the scope of this paper to extend the framework to two-sided trading algorithms that raise the issue of potential market manipulation \cite{alfonsi2012order,gatheral2012transient}.
\end{remark}


\subsection{Execution Costs in Limit Order Books}
Before proceeding with solving \eqref{oppr}, we discuss the connections between our setup and trading in electronic LOB markets that currently dominate equity trading in US and Europe. Trading in a LOB leaves a double footprint, both in space and in time. In space, an executed sell order consumes the matching standing limit orders on the bid side. (In our framework since there is no fill risk all trades are assumed to be market orders.) This shortens the respective queues on the bid side of the book and hence can move the best-bid. This is the aforementioned price impact represented by $g(\dot{x})$. Temporally, the executed order is recorded by other participants on the exchange message ticker affecting the observed order \emph{flow}. The effect is both direct (the immediate fact that a sell order of $\dot{x}_s$ shares was executed), and indirect (the fact that market participants adjust their statistical view of the order flow over time). Thus, to the extent that market participants monitor the ticker (rather than just observe the snapshots of the LOB), an order generates informational footprint. There is a lot of anectodal evidence that many HFT algorithms indeed track the time series of orders placed (for example to detect temporal trading patterns) and hence will react to this footprint. This implies that information costs are at least as important as instantaneous liquidity consumption.

Modelling order flow remains in its infancy. Indeed, while spatially the LOB can be easily summarized as a collection of queues (modeled via say the depth function \cite{alfonsi2010optimal}), the time series of the exchange ticker are much more complicated as market participants process multiple streams of information. There are both executed trades (i.e.~trades triggered from market orders) and limit orders, which themselves can be added, cancelled, or modified in other ways depending on the exchange. Orders further carry volume, time stamp and possibly participant type (and limit orders can be entered at any level of the LOB). These multi-dimensional data is moreover coupled in nontrivial ways, with temporal links both within series (e.g.~auto-correlation in the inter-order durations) and across series (e.g.~executed market orders tend to increase the arrival rate of limit orders at the touch). See for example a recent model of Cartea et al.~\cite{cartea2011buy} who fitted cross-exciting Hawkes processes to the basic order flow at the best-bid and best-ask queues. Given this challenges, our proposal to model order flow by a Ornstein-Uhlenbeck process is a compromise that will obviously need more adjustments before calibration to actual data. In Section \ref{sec:empirical-flow} we return to this question with some empirical stylized facts about order flow.

\begin{remark}
In this paper we focus on the temporal order flow and its associated imbalance that was conjectured by Easley et al.~\cite{easley2011microstructure,easley2012flow,easley2012optimal} to be related to market toxicity.
Other authors, such as \cite{cont2012price}, have used the terminology of \emph{order flow imbalance} (or just order imbalance) for a different object, namely ``spatial'' order imbalance. Namely, motivated by queueing notation, they mean the net difference between standing limit orders at the best-bid and best-ask. As shown by \cite{cont2012price} and \cite{Donnelly14}, order imbalance is predictive of the next price move (i.e.~correlated with the probability of the next price to be an up-tick or a down-tick) and is closely monitored by most HFT algorithms. While the LOB depth is related to the history of submitted limit orders, the relationship is highly complicated (due to shifting mid-price, hidden orders, etc.). Consequently, our order flow imbalance is not meant to be tied directly to the LOB depth or any immediate LOB properties, but rather provide a temporal summary of recent orders submitted.
%
%
\end{remark}

\subsection{HJB Formulation}
To minimize (\ref{oppr}) we adopt the standard stochastic control approach, utilizing the dynamic programming principle and Hamilton-Jacobi-Bellman (HJB) PDE.  Within this framework strategies are defined by their rate of selling, $\alpha_t :=-\dot{x}_t$ and the class of admissible strategies $\mathcal{A}(x)$ consists of all nonnegative $(\F_t)$-progressively measurable processes $(\alpha_t)_{0\leq t \leq T_0}$ for which
\begin{equation*}
x^{\alpha}_t := \left(x-{\int_{0}^{t}\alpha_s \, ds} \right)_+, \quad 0 \leq t,
\end{equation*}
belongs to $\mathcal{X}(x)$.  The value function of our problem can be expressed as
\begin{align}\label{hjb V}
v(x,y) = \inf_{(\alpha_t) \in \mathcal{A}(x)} \E_{x,y} \left[\int_{0}^{T_0} \left(g(\alpha_s) + \kappa Y^2_s+\lambda(x_s^\alpha) \right)\, ds \right].
\end{align}
If it exists, we define the corresponding optimal strategy as $\alpha^*(x,y)$. For each path of underlying $Y^0_t$, $\alpha^*$ induces the realized execution horizon $T_0(x,y) = \inf \{ t: x^{\alpha^*}_t = 0\}$ which is a random variable taking values in $[0,\infty)$.


Standard arguments  suggest that the value function $v(x,y)$ will satisfy a Hamilton-Jacobi-Bellman equation of the form
\begin{align}\label{eq:hjb}
0 =  \frac{1}{2} \sigma^2 v_{yy}-\beta y v_y  + \kappa y^2 + \lambda(x)+ \inf_{\alpha \ge 0} \{ g(\alpha) - \alpha v_x  - \phi(\alpha) v_y\},
\end{align}
with the boundary condition $v(0,y) = 0$ for all $y$. We observe that \eqref{eq:hjb} is a nonlinear parabolic PDE in $(x,y)$ for which the corresponding theory  (e.g.~regarding existence of classical solutions) is rather limited.

For the remainder of the section we concentrate on the linear information leakage and quadratic price impact case $g(\alpha) = \alpha^2$, $\phi(\alpha)=\eta \alpha$. In that situation, the candidate optimizer in feedback form is
\begin{align}\label{eq:alpha}
\alpha^*(x,y)=\frac{v_x+\eta v_y}{2}.
\end{align}
Substituting this feedback control into the PDE \eqref{eq:hjb} we have
\begin{align}\label{eq:hjb2}
0 =  \frac{1}{2} \sigma^2 v_{yy}-\beta y v_y  + \kappa y^2 + \lambda(x)- \left(\frac{v_x+\eta v_y}{2}\right)^2 .
\end{align}

Due to the state dependence of the class of admissible strategies $\mathcal{A}(x)$, the problem \eqref{hjb V} is a finite-fuel control problem.  As a result, there does not appear to be a tractable closed form solution which satisfies the zero boundary condition along $x=0$.  In Appendix \ref{App:AppendixA1} we illustrate a relatively straightforward method for solving \eqref{eq:hjb2} numerically via a finite difference scheme.

To understand the feedback strategy in \eqref{eq:alpha}, we pause to consider the derivatives $v_x$ and $v_y$. As we will see, $v_x$ is always positive but $v_y$ can be either positive or negative. Consequently, the candidate in \eqref{eq:alpha} may fail to be non-negative.

\begin{lemma}
The map $x \mapsto v(x,y)$ is strictly increasing for any $y$.
\end{lemma}

\begin{proof}
Fix $x < x' = x+\eps$ for a strictly positive $\eps$ and consider an ($\eps$-optimal) strategy $\alpha^\eps$ for $v(x',y)$. Let $T_x := \inf \{t :  x'_t = \eps \}$ be the random period to sell $x$ shares using $\alpha^\eps$. Then by absolute continuity of $t \mapsto x'_t$, $T_x < T_0(x',y)$. Moreover, $\alpha'(x,y) := \alpha^\eps_t(x',y) 1_{\{t \le T_x\}}$ is an admissible strategy for the initial conditions $(x,y)$ since it liquidates exactly $x'-\eps=x$ shares. Using the fact that $\alpha'$ is sub-optimal for $v(x,y)$ and that the second and third terms in \eqref{oppr} are strictly positive almost surely, we find $v(x,y) \le v(x',y; \alpha') < v(x',y)$.
\end{proof}



In \eqref{eq:hjb2}, the horizon is indefinite and ultimate liquidation is only modelled through the boundary condition. Thus, understanding the realized execution horizon $T_0(x,y)$ is only possible implicitly. Moreover, the nonlinearities in \eqref{eq:hjb2} make analysis intractable. To achieve tractability we  consider an approximate two-stage procedure. Thus, we first fix a horizon $T$ by imposing the constraint $T_0 = T$. We then solve the resulting fixed-horizon problem to
find the best strategy $\alpha^*(T,x,y)$ and value function $v(T,x,y)$. In the second step, we optimize over $T$, to find the statically optimal horizon $T^*(x,y)$. Finally, we build the semi-dynamic strategy $\talph(x_t,y_t)=\alpha^*( T^*( x_t,y_t), x_t,y_t)$. Thus, $\talph$ recomputes $T^*$ as the state variables $(x_t,y_t)$ evolve and uses the corresponding static trading rate. This approach is analogous to the receding-horizon setup in nonlinear control \cite{Primbs07}. Indeed, the initial use of $\alpha^*( T^*( x,y), x,y)$  at $t=0$ corresponds to model predictive control and $\talph(x_t,y_t)$ then continuously rolls the initial condition because of the stochastic fluctuations encountered. The above plan is implemented in Sections \ref{sec:linear-quadratic} and \ref{sec:optimal-T} respectively. In the latter section we also compare the execution trajectories and resulting costs from the various strategies.


\section{Linear Quadratic Setup on Finite Horizon}\label{sec:linear-quadratic}
Fix $T < \infty$. We consider the analogue of \eqref{eq:hjb} on $[0,T]$.   To avoid confusion we let $u$ denote the value function when defined on the fixed horizon:
\begin{equation}\label{vfun}
u(T,x,y)=\inf\limits_{(\alpha_t) \in \mathcal{A}(T,x)}\E_{x,y}\left[\int_{0}^{T}\alpha^{2}_s+\kappa Y^{2}_s+\lambda(x_s^\alpha) \, ds\right].
\end{equation}
For expository purposes, we use the time-to-maturity parametrization for $u$ so that the first argument $T$ represents time \emph{until} the deadline.
Strategies on $[0,T]$ are defined in similar fashion to those in Section \ref{sec:model} but with a constraint $x_T = 0$ at the terminal time $T$. Forced liquidation by $T$ is achieved by leveling an infinite penalty if not completed, leading to a singular initial condition of the form
\begin{equation}\label{eq:singular}
\lim\limits_{T\downarrow 0} u(T,x,y)=\begin{cases}
0 & \text{if }x=0 \\
+\infty & \text{if }x\neq 0.
\end{cases}
\end{equation}

Note that the second term in \eqref{vfun} continues accruing until the terminal time regardless of whether the liquidation is completed earlier. Therefore, instead of the boundary condition $v(0,y) = 0$ we have (due to zero being absorbing for $x_t$) $u(t,0,y) = \E_{0,y}[ \int_0^t \kappa (Y^0_S)^2 \,ds]$; see Lemma \ref{cor:expect} for explicit formula for the latter expression.

To obtain explicit solutions to \eqref{vfun}, the next section treats the case in which $\phi$ is independent of $\alpha$ over a fixed horizon.  In other words, the trader may impact the order flow process, but impact can be modelled in a deterministic fashion. Section \ref{sec:dynamic} then addresses the proportional footprint $\phi(\alpha)=\eta \alpha$ case, still over fixed time horizon.    It will be shown that the strategies obtained in Sections \ref{sec:myopic}-\ref{sec:dynamic} are not too suboptimal compared to the indefinite-horizon model laid out in Section \ref{sec:model}.

\subsection{Myopic Execution Strategies}\label{sec:myopic}
In classical optimal execution models \cite{alfonsi2010optimal}, \cite{almgren2001optimal} and \cite{almgren2003optimal}, optimal execution rates are deterministic, i.e.~$\alpha_t$ is pre-determined. In this scenario, informational costs would also be deterministic. Therefore, we examine the case where $\phi(\alpha)$ is independent of $\alpha$ (but possibly depends on time $t$). Under that assumption, we can separate the two terms in \eqref{oppr} since the dynamics of $(Y_t)$ are not directly affected by the trader; the performance criterion simplifies to
\begin{equation}\label{simp}
\inf_{(x_t) \in \mathcal{X}(x)} \left( \int_0^T  \dot{x}_s^2 + \lambda( x_s) ds \right) + \int_0^T  \kappa \E_y[ Y_s^2] ds.
\end{equation}
Because $(Y_t)$ is independent of the control $\alpha$, optimal strategies are defined only by the first term in \eqref{simp}. Consequently, the resulting $(\alpha_t)$ is independent of $Y_t$ and hence $t \mapsto x^*_t$ is deterministic.
Thus, strategies based on \eqref{simp} are myopic in the sense that they entirely ignore the potential ``footprint" left by the trader's actions, instead focusing solely on instantaneous cost and inventory risks. The following Lemma, proven in Appendix \ref{App:AppendixA} provides the solution to \eqref{simp} for popular choices of inventory risk.

\begin{lemma}\label{lem:exec-curve}
Consider the calculus-of-variations problem of finding
$$
\cI(T,x) := \inf_{(x_t)} \int_0^T  (\dot{x}_t^2 + \lambda(x_t)) \, dt
 $$
where the minimization is over all absolutely continuous curves $t \mapsto x_t$ with $x_0 = x$, $x_T = 0$ and under the constraint that $x_t$ is non-increasing. Then the optimal ``myopic'' strategies (with $\alpha_t \equiv -\dot{x}_t$) are
\begin{align} \label{eq:x-star-lin}
\left\{ \begin{aligned}
 x^{ML}_t & = \frac{x (T-t)}{T}; \\
\alpha^{ML}_t & =  \frac{x}{T}; \\
 \cI^{ML}(T,x) & =  \frac{x^2}{T}, \end{aligned}\right\}
  &  \quad \text{ if } \lambda(x) = 0;\\
\label{eq:x-star-hyp}
\left\{
\begin{aligned}
x^{MH}_t & =  \frac{x \sinh(\sqrt{c}(T-t))}{\sinh(\sqrt{c} T)}; \\
\alpha^{MH}_t & = \frac{\sqrt{c} x \cosh( \sqrt{c}(T-t) )}{\sinh(\sqrt{c} T)}; \\
\cI^{MH}(T,x) & = \sqrt{c} x^2 \coth(\sqrt{c} T ),  \end{aligned} \right\} & \quad \text{ if } \lambda(x) = c x^2;
\end{align}
\begin{align}\label{eq:x-star-quad}
\left\{ \begin{aligned}
x^{MQ}_t & =  \left(\frac{c t^2}{4}-t\left(\frac{c\hat{T}}{4}+\frac{x}{\hat{T}}\right)+x \right) \mathbbm{1}_{\{ t<\hat{T}\}}; \\
\alpha^{MQ}_t & = \left(\frac{c \hat{T}}{4}+\frac{x}{\hat{T}}-\frac{c t}{2}\right)\mathbbm{1}_{\{ t<\hat{T}\}}; \\
\cI^{MQ}(T,x) & = \left(-\frac{c^2 \hat{T}^3}{48}+\frac{c \hat{T} x}{2}+\frac{x^2}{\hat{T}}\right);
 \\
& \text{where }\; \hat{T}  := \min(T, \frac{2\sqrt{x}}{\sqrt{c}}),
\end{aligned} \right\}
 & \quad \text{ if } \lambda(x) = c x.
\end{align}
\end{lemma}
The superscripts ${ML}, {MQ}, {MH}$ respectively stand for the Myopic Linear, Quadratic and Hyperbolic models.
The first case ML yields linear selling and the TWAP strategy, or if time is parametrized in volume time, the classic VWAP trading strategy. The second strategy $x^{MH}_t$ and its corresponding rate $\alpha_t^{MH}$ represents exponential selling, and is the optimal strategy presented in the original Almgren-Chriss model \cite{almgren2001optimal}. This risk term results from the trader's effort to minimize the variance of liquidation cost.  From the perspective of an inventory risk measure, one natural alternative is $\lambda(x_s) = c x_s$, which has the attractive property of being proportional to value-at-risk and results in a selling strategy $x^{MQ}_t$ that is quadratic in $t$.  Yet as explained for a similar problem in \cite{gatheral2011optimal}, buying might result with position size small relative to $T$. Imposing the
constraint that $x^{MQ}_t$ is decreasing then leads to the modified solution \eqref{eq:x-star-quad} which in the case $\hat{T} < T$ causes liquidation to end prior to the terminal time $T$.

  \begin{figure}[ht]

  \centering
 \includegraphics[scale=.55, trim=0.1in 0.1in 0.1in 0.2in]{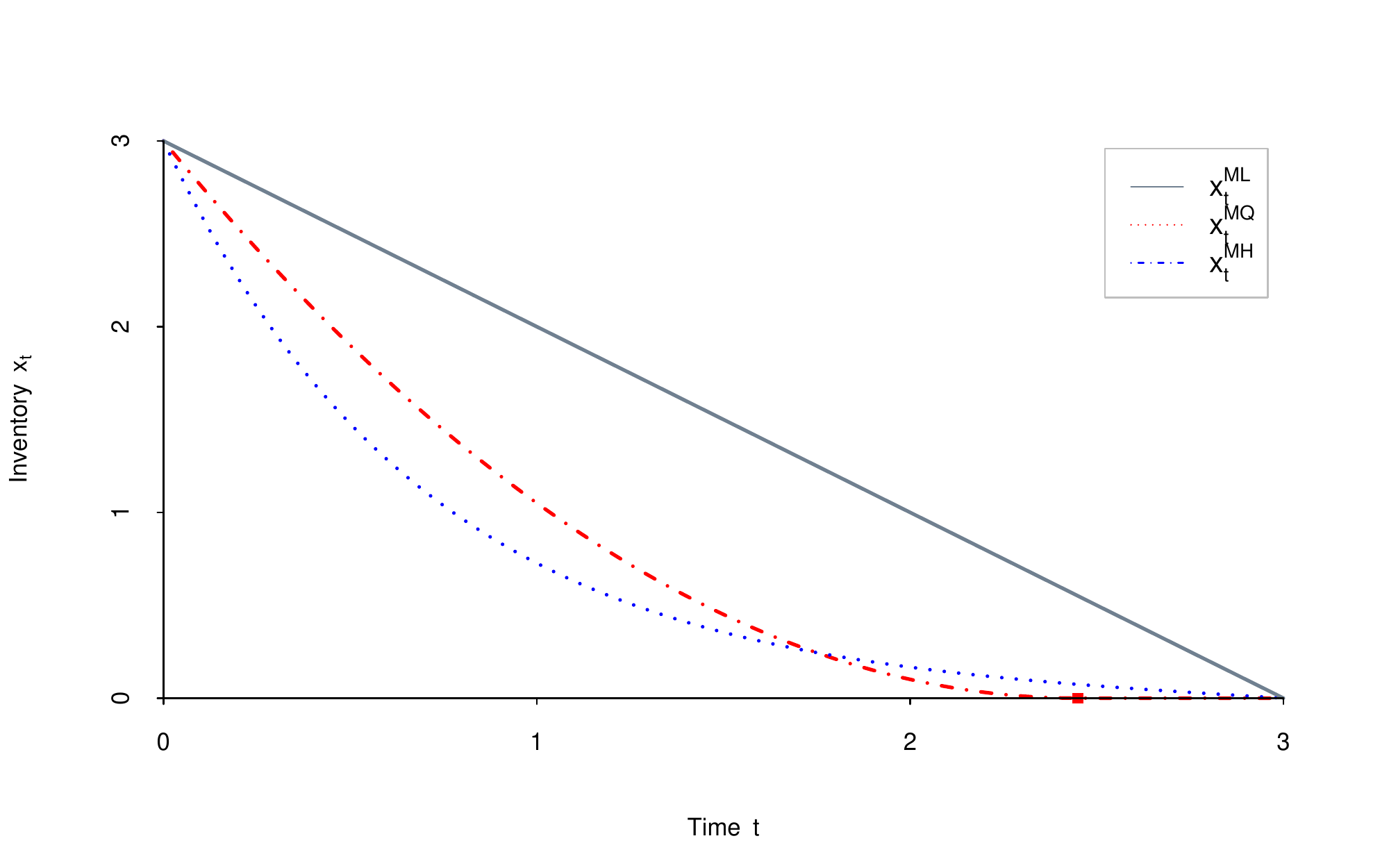}
 \caption{Optimal trajectories from Lemma \ref{lem:exec-curve}. Figure drawn for initial inventory $x=3$, horizon $T=3$, and $c=2$. This leads to $\hat{T}= \frac{2\sqrt{x}}{\sqrt{c}}=2.45$ in the quadratic scenario $x^{MQ}$ of \eqref{eq:x-star-quad}.\label{fig:exec-curves}}
  \end{figure}

Figure \ref{fig:exec-curves} details  the three execution curves described in Lemma \ref{lem:exec-curve} for $x = 3$ and $T=3$. Compared to the VWAP strategy $x^{ML}$, the non-zero inventory risk terms in strategies $x^{MH}$ and $x^{MQ}$ lead to higher rates of selling initially. The execution rate in $x^{MH}$ is proportional to $x$, cf. \eqref{eq:x-star-hyp} and liquidation occurs exactly at $T$. In contrast, as $x^{MQ}_t$ becomes small, the linear risk term $c x^{MQ}_t$ becomes more punitive and it may be optimal to end liquidation prior to time $T$.  Figure \ref{fig:exec-curves} shows inventory $x^{MQ}_t$ reaching $0$ at time $\hat{T}=\sqrt{6}$ as defined in \eqref{eq:x-star-quad}.

 We now turn our attention to the expected execution costs which arise from the order flow imbalance.
   Fixing $(\alpha^*_t)$, we can  view the corresponding information impact $\phi(\alpha^*_t)$ also as a deterministic function of $t$, allowing direct evaluation of the second term in \eqref{simp} using the explicitly available Gaussian distribution of  $Y_t$.

 \begin{lemma}\label{cor:expect}
 Given a deterministic, time-dependent flow impact $\phi(\alpha_t) = \phi_t$, and $Y_0=y$, $Y_t$ has the second moment
 \begin{equation}\label{eq:expY2}
 \E_y \left[ Y_t^2 \right]  = \mu^2_t+\sigma^2_t,
 \end{equation} where
 \begin{equation}\label{eq:params} \left\{
 \begin{aligned}
 \mu_t &=y e^{-\beta t}-\int_{0}^{t} e^{-\beta(t-s)} \phi_s \, ds, \\
 \sigma^2_t &=\frac{\sigma^2}{2\beta}(1-e^{-2\beta t}).
 \end{aligned} \right.
 \end{equation}
 \end{lemma}
 Lemma \ref{cor:expect} can be verified by directly solving the SDE \eqref{dyn} and recalling the Gaussian distribution of $Y_t$.


 Putting everything together we obtain the expected total cost for the family of myopic execution strategies in Proposition \ref{lemma1}.  We reiterate that while realized costs instantaneously depend on the stochastic process $(Y_t)$, strategies in this section are purely deterministic and do not adapt to $(Y_t)$. The solutions below are labeled according to the form of $\lambda(x)$; while formally the two terms in \eqref{simp} are decoupled, it is of course logical to match the resulting solution $\cI$ of the instantaneous price-impact part with the corresponding expectation of $\E_y[ \int_0^T Y^2_s \, ds]$, which is the convention we follow in Proposition \ref{lemma1}.

 \begin{proposition}\label{lemma1}
 Suppose that $\phi_t(\alpha) = \phi_t$.
 The corresponding value function $u$ is given by
 \begin{equation}\label{eq:generalV}
 u^M(T,x,y)=\cI(T,x) + \mathcal{O}(T,x,y),
 \end{equation}
 where $\cI$ is defined in \eqref{eq:x-star-lin}-\eqref{eq:x-star-quad} and
 $\mathcal{O} = \displaystyle{  \kappa \int_{0}^{T} (\mu^2_t + \sigma^2_t) \, dt }$ from Lemma \ref{cor:expect} is
 %
 \begin{align}\label{eq:val1}
 \cO^{0}(T,x,y) & =\frac{\kappa y^2}{2\beta} \left(1-e^{-2\beta T}\right) + \frac{\kappa \sigma^2}{4\beta^2}\left(
  2\beta T+e^{-2\beta T}-1\right) \quad\text{if } \phi_t \equiv 0; \\
\label{eq:val2}
 \cO^{ML}(T,x,y)& =\cO^{0}(T,x,y) + \frac{ \kappa \eta x y}{\beta^2 T}\left(2e^{-\beta T}-1-e^{-2\beta T}\right) \\ \notag
  & \quad + \frac{\kappa \eta^2 x^2 }{2\beta^3 T^2}\left(2\beta T+4e^{-\beta T}-e^{-2\beta T} -3\right) \qquad\text{if }\phi_t = \eta\alpha^{ML}_t
 \end{align}
Closed form expressions are also available for $\cO^{MQ}$ and $\cO^{MH}$, see Appendix \ref{App:AppendixB}.
 \end{proposition}

Thus, the overall cost of liquidation has two components: the $\cI(T,x)$ term that depends only on $(T,x)$, and the informational footprint term $\cO(T,x,y)$ that also depends on $y$.  In terms of $x$, $\cO$ is constant if $\phi_t=0$, linear if $\phi_t = \eta\alpha^{ML}$, and quadratic otherwise.  As a function of $y$, $\cO$ is quadratic thanks to the linear dynamics of $(Y_t)$ and quadratic informational cost. Financially, the $y^2$ term represents higher costs due to trading in an unbalanced market, while the $y$-term adjusts to the fact that selling in a market dominated by buyers is favorable to competing with other sellers for scarce liquidity. The following Corollary shows that the ``best'' level of $y$ is positive (or $0$ if $\phi_t=0$). Intuitively, it is best to begin trading in an environment with positive order flow so that the trader's selling activity pushes the order imbalance towards $0$ and reduces informational costs.

\begin{corollary}\label{ymin}
 Suppose that $\phi_t(\alpha) = \phi_t \ge 0$.  Then the flow imbalance that minimizes expected execution cost is non-negative, $\arg\min_y\left\{ u(T,x,y) \right\} \geq 0$ for any $T,x$.
 \end{corollary}
\begin{proof}
As already discussed, $u^{M}(T,x,y)$ is quadratic in $y$ and the coefficient of $y^2$ is always as in \eqref{eq:val1}. By inspection it is positive. The dependence on $y$ comes from the $\cO$ terms that are of the form
$$
\cO(T,x,y) = \kappa \int_0^T (y e^{-\beta t} - A_t)^2 + \sigma_t^2 \,dt
$$
where $A_t \ge 0$ (strictly positive as soon as $\phi_t > 0$ on an interval of positive measure, cf.~\eqref{eq:def-A}). It follows that the coefficient of $y^2$ in $u^{M}(T,x,y)$ is $\kappa \int_0^T e^{-2\beta t} dt > 0$ and of $y$ is $ -\int_0^T 2 \kappa e^{-\beta t}A_t dt \le 0$.
%
Thus, setting $\partial_y u^{M}(T,x,y)=0$ and solving for $y$ yields a non-negative result.
\end{proof}

 \subsection{Dynamic Execution Strategies}\label{sec:dynamic}
 We now return to the problem in \eqref{vfun}, letting $\phi(\alpha_t) = \eta \alpha_t$.  The optimal dynamic strategy is adapted to the order flow imbalance process $(Y_t)$ and the trader's rate of selling directly influences $(Y_t)$.  To avoid confusion we will denote dynamic strategies and expected costs by $\alpha^D_t$ and $u^D$, respectively.
 The HJB PDE for $u^D(T,x,y)$ is
 \begin{equation}\label{fullpde}
 u^D_T=\frac{1}{2}\sigma^2 u^D_{yy}-\beta y u^D_y+\kappa y^2 + \lambda(x)+\inf_{\alpha \ge 0} \left\lbrace g(\alpha) -\alpha u^D_x-\eta \alpha u^D_y\right\rbrace,
 \end{equation}
 with $u^D(0,x,y) = +\infty$ unless $x=0$.  Note that \eqref{fullpde} is identical to \eqref{hjb V} but for the time derivative on the left hand side of the equation which is introduced due to the time-dependence arising from the constrained horizon $T$.  Assuming $g(\alpha) = \alpha^2, \lambda(x)=cx^2$,  and inserting the feedback control as in \eqref{eq:hjb2} yields a semi-linear, parabolic PDE
 \begin{equation}\label{eq:fpde2}
 u^D_T=\frac{1}{2}\sigma^2 u^D_{yy}-\beta y u^D_y+\kappa y^2 + c x^2-\left(\frac{u^D_{x}+\eta u^D_y}{2}\right)^2.
 \end{equation}
 with initial condition \eqref{eq:singular}.  To obtain \eqref{eq:fpde2}, we let $\alpha$ be unconstrained and allowed to become negative.  This allows us to find the following candidate solution by exploiting the linear-quadratic structure. The motivation  comes from $u^{M}$ in Proposition \ref{lemma1} where we find a similar result: quadratic in $x$ and $y$ with an $xy$ term that adds additional costs when $y<0$.  
 \begin{proposition}\label{lem:riccati}
The solution of \eqref{eq:fpde2} has the form
 \begin{equation}\label{eq:riccati}
 u^{DH}(T,x,y)=x^2 A(T)+y^2 B(T)+x y C(T)+ x D(T) + y E(T) + F(T),
 \end{equation}
 where    $D(T)=E(T)\equiv 0$, $A,B,C,F$ solve the matrix Riccati ordinary differential equations (ODEs)
 \begin{equation}\label{eq:syste}
 \begin{cases} A^{\prime}(T) &= -A^2-\eta AC-\frac{\eta^2 }{4}C^2+c\\
 B^{\prime}(T) &=-\eta^2 B^2-B(\eta C+2\beta)+\kappa-\frac{1}{4}C^2\\
 C^{\prime}(T) &= -\frac{\eta}{2}C^2-C(\eta ^2 B+A+\beta)-2\eta AB\\
 F^{\prime}(T) &=  \sigma^2 B,\end{cases}
  \end{equation}
and we have the following initial conditions
  \begin{equation}\label{eq:cond}
  \begin{cases} \lim\limits_{T \downarrow 0} A(T)= +\infty \\
  B(0) =   C(0) =  F(0) = 0. \end{cases}
   \end{equation}
 The optimal rate of liquidation is
  \begin{equation}\label{eq:a-star3}
  {\alpha}_t^{DH}(T-t,x_t,Y_t)= \frac{x_t(2A(T-t)+\eta C(T-t))+Y_t(C(T-t)+2\eta B(T-t))}{2}.
  \end{equation}
 \end{proposition}

See Appendix \ref{App:AppendixC} for proof of Proposition \ref{lem:riccati}.
 We reiterate that in \eqref{eq:syste} $A,B,C,F$ are functions of time remaining and that we have simplified the notation by omitting the time argument (i.e.~$A^\prime = A^\prime (T)$, etc.) on the right side of \eqref{eq:syste}.   Close to the deadline $T$, the impact from impacting $Y$ disappears, and \eqref{eq:fpde2} converges to the myopic linear case of \eqref{eq:x-star-lin}. This can be observed by formally linearizing the Riccati system \eqref{eq:syste} in the regime $T-t = \eps$
and using the initial conditions \eqref{eq:cond}. We obtain the following expansions in $\eps$:
 \begin{equation}\label{eq:shortsyste}
  \begin{cases} A(\epsilon)= \frac{1}{\epsilon} + O(\eps)\\
  B(\epsilon)=\kappa \epsilon + O(\eps^2); \\
  C(\epsilon) = -\eta \kappa \epsilon + O(\eps^2); \\
  F(\epsilon)= \frac{\sigma^2 \kappa \epsilon^2}{2} + O(\eps^3).\end{cases}
   \end{equation}
Inserting into \eqref{eq:a-star3} gives the short term trading rate $\alpha^{DH}(\epsilon,x_t,Y_t)=\frac{x_t}{\epsilon} +O(\epsilon)$. This heuristically confirms that the strategy \eqref{eq:a-star3} is admissible which can also be observed in Figure \ref{fig:rates-comp} below: as $t \to T$, the dynamic trading rate stabilizes, resembling a VWAP strategy.


 \begin{remark}\label{rem:linear-riccati}
{It is also possible to set-up and solve linear quadratic problems for other functional forms of inventory risk $\lambda(x)$. In the linear case  $\lambda(x) = c x$, the resulting Riccati system will have non-zero linear coefficients $D(T)$ and $E(T)$ of $x$ and $y$. However, dynamically satisfying the constraint $x \geq 0$ is not tractable (cf.~$\hat{T}$ in \eqref{eq:x-star-lin}) and we found that the resulting unconstrained strategies tend to lead to wild buying-and-selling.

 In the constant case $\lambda(x)=c$ the Riccati system is almost the same as \eqref{eq:syste} (again $E(T) = D(T) \equiv 0$) except that the $c$ term moves to the fourth line: $F_{DL}'(T) =  \sigma^2 B(T) + c$. In the sequel we use both of these solutions for the numerical illustrations.}
%
 \end{remark}


 The equations in \eqref{eq:syste} can be dealt with using a software package such as \texttt{R}.  It is however necessary to replace the singular initial condition \eqref{eq:singular} with the condition
 \begin{equation}\label{eq:M}
 \lim\limits_{T\downarrow 0} u^D(T,Y,x)=\begin{cases}
 0 & \text{if }x=0 \\
 M & \text{if }x\neq 0
 \end{cases}
 \end{equation}
 for a constant $M$ large, essentially allowing for a non-zero position at time $T$ which must then be liquidated in a single order at some additional cost. This is equivalent to introducing a boundary layer $[0,\eps]$ and solving on $T \in [\eps, \infty]$ whereupon $M=1/\eps$ is the right choice based on \eqref{eq:shortsyste}.

 The optimal trading rate $\alpha^D$ in \eqref{eq:a-star3} is linear in both $x_t$ and $Y_t$. The former feature is similar to the hyperbolic situation in \eqref{eq:x-star-hyp} where $\alpha^{MH}_t$ is also linear in $x_t$.  We next illustrate how the dynamic strategy $\alpha_t^D$ compares to its myopic counterpart $\alpha_t^{M}$.  With fixed terminal time $T$, the incentive for the trader to speed up or slow down under strategy $\alpha_t^D$ arises from the trader's desire for more balanced order flow.  Note that there is no incentive to accelerate one's trading in order to complete the task and exit the market prior to time $T$ since costs from $Y$ accrue until $T$.  For positive order flow imbalance, trading more quickly in the present results in lower execution costs in the future because $(Y_t)$ will be closer to $0$ as a result of his activity. Likewise, if order flow imbalance is negative, it is better to reduce trading speed so as not to pull $(Y_t)$ further from $0$.  For negative $Y$ and large enough $T$ (or large enough $\kappa$, $\beta$), $\alpha_t^D$ may become negative (i.e.~it may be optimal to begin buying), however this happens only under extreme parameters.


 Figure \ref{fig:dyn-vs-myo} illustrates the results of Proposition \ref{lem:riccati} for a simulated path of $(Y^0_t)$ comparing the myopic $\alpha^{MH}$ versus the adaptive $\alpha^{DH}$. As can be observed, both strategies have a broadly similar shape, with $\alpha^D$ ``fluctuating'' around $\alpha^{MH}$. We also observe that $\alpha^{D}$ is less aggressive initially, starting out slower and then speeding up (relative to $\alpha^{MH}$) after $t > 1.5$.

   \begin{figure}[ht]
   \centering
  \includegraphics[height=2.5in, trim=0.1in 0.2in 0.1in 0.1in]{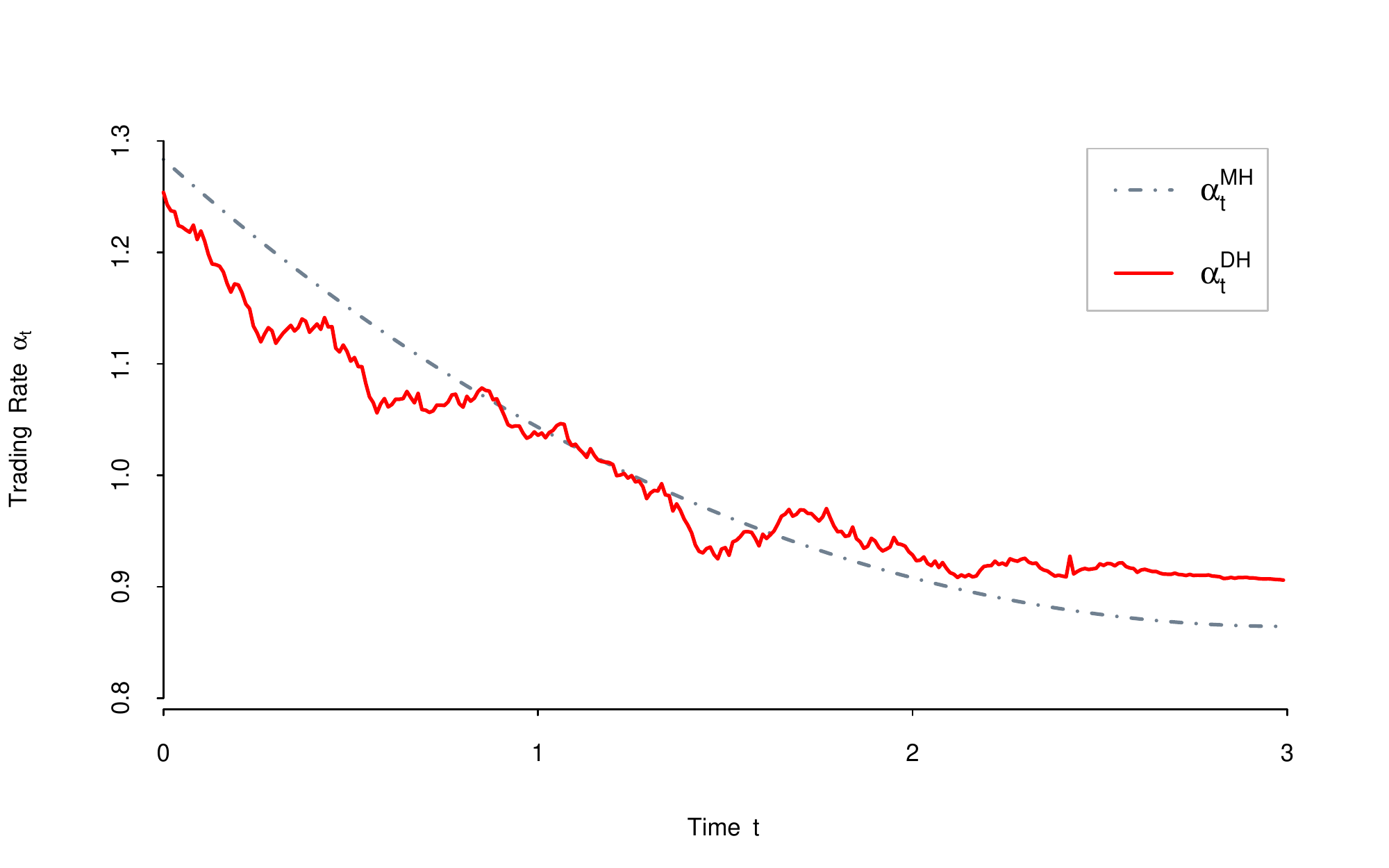}
  \caption{Trading rates $\alpha_t^{DH}$ and $\alpha_t^{MH}$ for a sample simulated path of $(Y_t)$.  The figure is drawn for parameter values $T=3$, $\kappa=10$, $\sigma=.14$, $\beta=.05$, $\eta=.05$, $\lambda(x)=0.1 x^2$, and initial condition $x_0 =3, Y_0=0$. \label{fig:dyn-vs-myo}}
   \end{figure}

 \section{Optimizing Execution Horizon}\label{sec:optimal-T}
 We now move to the second step of the approximate solution scheme and remove the fixed horizon constraint.  Given $u(T,x,y)$, define
 \begin{equation*}
 T^* := \arg\min_T\ u(T,x,y) .
 \end{equation*}
The next Lemma shows that  $T^*$ is finite in all the cases considered so far.

 \begin{lemma}\label{lem:T-dependence}
 For any fixed $x$, there exists $\bar{T}$ such that $\partial_T u(T,x,y) > 0$ for all $T > \bar{T}$ and all $y$.
 \end{lemma}

 \begin{proof}
 Recall that $u = \cI + \cO$, cf.~\eqref{eq:generalV}.
 As $T \to \infty$, the variational problem \eqref{simp} for $\cI$ becomes independent of $T$. By inspection, $\lim_{T \to \infty} \cI^{ML}(T,x) = 0$ and $\lim_{T \to \infty} \cI^{MH}(T,x) = \sqrt{c}x^2$. In the quadratic case, for $T$ large enough, $\hat{T} = \frac{ 2 \sqrt{x}}{\sqrt{c}}$ so that $\lim_{T \to \infty} \cI^{MQ}(T,x) = \frac{4}{3}x^{3/2}\sqrt{c}$ for some function of the initial inventory. In contrast, $\cO(T,x,y) = \int_0^T \mu_t^2 + \sigma_t^2\,dt$ grows at least linearly in $T$ since $\lim_{t\to\infty} \sigma_t^2 = \frac{\sigma^2}{2\beta}$. Also, the first term is non-negative and it follows that $ \partial_T \cO(T,x,y) \ge \frac{\sigma^2}{2\beta}$ asymptotically as $T \to \infty$ and for any $y$. Hence, $\partial_T u > 0$ for all $T$ large enough.\hfill
 \end{proof}

 Figure \ref{fig:t-dep} illustrates Lemma \ref{lem:T-dependence} for the dynamic hyperbolic strategy with value function $u^{DH}$. We observe that $u^{DH}(T,x,y)$ appears to be convex in $T$ with a unique global minimum $T^*(x,y)$. Moreover, $T^*(x,y)$ is largest for $y \simeq 0$ and smallest for negative $y$. This matches the intuition that trading is slowest in balanced markets where informational costs are large, and fastest in sell-driven markets where further information leakage is minimal.

    \begin{figure}[ht]
    \centering
   \includegraphics[height=2.5in]{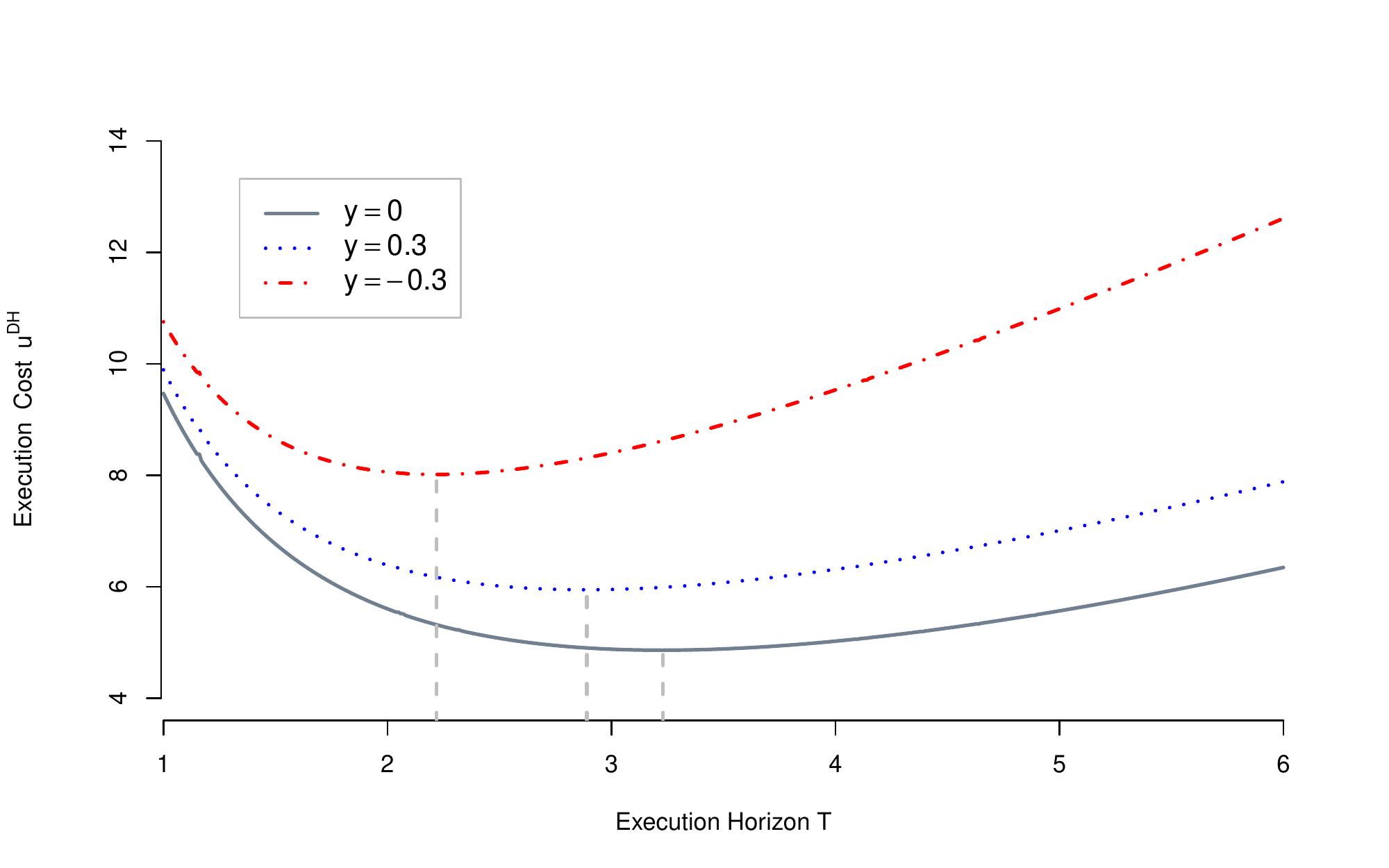}
   \caption{ Expected execution cost $u^{DH}(T,x,y)$ as a function of $T$ for different values of order flow imbalance $y$. The dashed line indicate the value of $T$ achieving the minimum.  Figure drawn for $\beta=.05$, $\sigma=.14$, $\eta=.05$, $\kappa=10$, $\lambda(x)=0.1 x^2$ and inventory $x=3$. \label{fig:t-dep} }
    \end{figure}

Given initial $(x,y)$ and corresponding $T^*(x,y)$, let $\alpha^{M}_t(T^*,x,y)$ (and similarly $\alpha^D_t(T^*,x,y)$) be the resulting strategy over the fixed horizon $[0,T^*)$. This provides a static or open-loop optimal execution strategy, since $T^*$ is fixed and not adjusted as order flow $Y_t$ changes. We can also construct a ``dynamic'' strategy by continually recomputing $T^*(x_t, Y_t)$ using the latest datum $(x_t,Y_t)$. We denote the latter as $$\talph^M_t(x,y) := \alpha^M(T^*(x_t,y_t),x_t,y_t)$$ with associated value function $\tilde{u}^{M}$.  The corresponding $\talph^D(x,y)$ and $\tilde{u}^D$ are defined in similar fashion. The approach of ``rolling'' the horizon $T^*$ as the underlying stochastic state changes is known as receding horizon control or model predictive control, see e.g.~\cite{Primbs07}.

\begin{remark}\label{rem:rebalance}
Recomputing $T^*$ can be done at any frequency. Namely, given a (stochastic) set $0 = t_0 < t_1 < \ldots$, one can construct the strategy $\alpha(\fT(t), x_t, Y_t)$ where $\fT(t) := T^*( x_{t_k}, Y_{t_k})$ and $t_k =\max\{ t_i : t_i < t\}$. For example, one can take $t_k = \inf \{ t: x_t \le (K-k)x/K\}$, giving $K$ rebalancing periods, during each of which $1/K$ of total inventory is liquidated.
\end{remark}


Before moving forward we pause briefly to summarize the various strategies that have been defined.  The fully dynamic strategy which solves the original indefinite-horizon control problem \eqref{eq:hjb} in Section \ref{sec:model} is denoted $\alpha^*(x,y)$.  In Section \ref{sec:myopic} we defined a family of myopic strategies on a fixed horizon, generally denoted $\alpha^M(T,x,y)$ and specific cases
addressed in Lemma \ref{lemma1}. Continually optimizing $T^*$ then yields the receding horizon strategy $\talph^M(x,y)$.  In Section \ref{sec:dynamic} we introduced the dynamic strategy $\alpha^D(T,x,y)$, which adapts to changing flow imbalance over a fixed horizon, as well as the corresponding receding $\talph^D(x,y)$.


We proceed to compare execution cost statistics across the described strategies. For easier interpretation we consider the case of zero inventory penalization, $\lambda(x)=c$ independent of $x$, so that the benchmark strategy (without informational costs) is VWAP, i.e.~constant trading rate.  The precise parameters were: timing risk $\lambda(x)=c=.1$ (i.e~linear timing costs), initial inventory and initial order flow imbalance $x=3$ and $y=0$ respectively, $\eta=.075$, $\kappa=10$, $\beta=.05$ and $\sigma=.14$. Against the original strategy $\alpha^*(x,y)$, we also compare the adaptive $\talph^D$ and $\talph^{ML}$. Both of these adjust the execution horizon by optimizing $T$ in the fixed-horizon solution. Practically, this was achieved by discretizing in time ($\Delta t=.01$) and recomputing $T^*(x_t, Y_t)$ at each time step, see Remark \ref{rem:rebalance}.
 Recall that $\alpha_t^{ML} = x_t/T^*(t)$. To understand the frequency of above ``rebalancing'', we also show results for the strategy ${\alpha}^{ML}_t(T^*_{(2)}, x,y)$ which recomputes the horizon midway through the liquidation process, when $x_t=\frac{x}{2}$.  The corresponding path of $x_t$ is therefore piecewise linear with two pieces, see Figure~\ref{fig:rates-comp}. This is a convenient compromise in the $VWAP$ setting and nicely illustrates the advantage gained when the trader is allowed to adjust the execution horizon. Finally, to understand the importance of adaptively adjusting $T^*$, we also compare to the static $\alpha^{DL}(T^*,x,y)$ and $\alpha^{ML}(T^*,x,y) = x/T^*(x,y)$.

 Table \ref{tab1} shows some summary statistics about the distribution of realized costs $J(\alpha) := \int_{0}^{T_0} \left(\alpha_s^2+ \kappa Y^2_s+ c \right) ds$. The results were produced with $2000$ simulated paths of $(Y^0_t)$.  The actual realized order flow imbalance paths for each strategy reflect the assumption that $\phi(\alpha_t)=\eta \alpha_t$ represents the true form of information leakage. Beyond the average expected costs $u(x,y) := \E_{x,y}[ J(\alpha)]$, we also report the standard deviation  $SD$ and quantiles $q_\cdot$ (at the 5\% and 95\% level) of $J(\alpha)$ which are important for risk-management perspective. Lastly, we also report the average realized horizon $\E[ T_0]$. Of course for non-adaptive strategies, $T_0 \equiv T^*(x,y)$ is constant.
  Comparing each $\talph$ to its respective fixed horizon counterpart demonstrates the importance of utilizing ``adaptive" execution horizon.  Similarly, comparing respective myopic to dynamic strategies shows that the modelling of deterministic information leakage in the former is not too suboptimal compared to the fully dynamic proportional information leakage strategy of the latter.  The cost improvements achieved through optimizing the horizon tend to dominate those obtained through adopting a dynamic strategy in lieu of a myopic strategy.
\begin{table}[ht]
\footnotesize
\begin{center}$$
\begin{array}{lllllll}
\hline
\multicolumn{7}{ c }{ \text{Optimal Execution Strategy}} \\
 & v & \tilde{u}^D & \tilde{u}^{ML}& u^{ML}_{(2)} & u^D & u^{ML}\\ \hline \hline
\E[J(\alpha)] & 4.257 & 4.264 & 4.317& 4.411 & 4.483 & 4.547 \\
SD(J(\alpha)) &1.50 &1.45 &1.39&1.49 &1.77 &1.84 \\
q_{.05}(J(\alpha)) & 2.70 & 2.76 & 2.83&2.96 & 3.11 &3.12 \\
q_{.95}(J(\alpha)) & 7.33 & 7.28 & 7.10& 7.50& 8.19&8.42 \\
\E[T_0] &3.87 &3.70 &3.48& 3.44& 3.43 & 3.43 \\ \hline
\end{array}$$
\end{center}
\caption{Statistics for six execution strategies including average realized cost $J(\alpha)$, standard deviation, $.05-$ and $.95-$quantiles of realized costs as well as average realized execution horizon.  
\label{tab1}}
\end{table}

Of particular interest is that the closed form solution computed for $u^{ML}$ and resulting strategy $\talph_t^{ML}$ form a reasonable approximation for the difficult indefinite horizon setup in \eqref{hjb V} - \eqref{eq:hjb2}.  The cost improvement of the fully dynamic $\alpha^*$ over VWAP strategy $\alpha^{ML}$ is approximately $6.8\%$.  This appears somewhat modest, but note that the latter strategy is applied to the horizon $[0,T^*]$, which is the statically optimal horizon computed at $t=0$ with the value function $u^{ML}$. 

Figure \ref{fig:rates-comp} illustrates the various strategies for a sample simulated order flow imbalance path $(Y^0_t)$.  The one-sided order flow in this particular simulation causes the ``adaptive" horizon strategies to accelerate trading and shorten the horizon relative to the fixed horizon strategies.

\begin{figure}[ht]
\begin{minipage}{0.95   \textwidth}
\includegraphics[width=1\textwidth,trim=0in 0.1in 0.1in 0.2in]{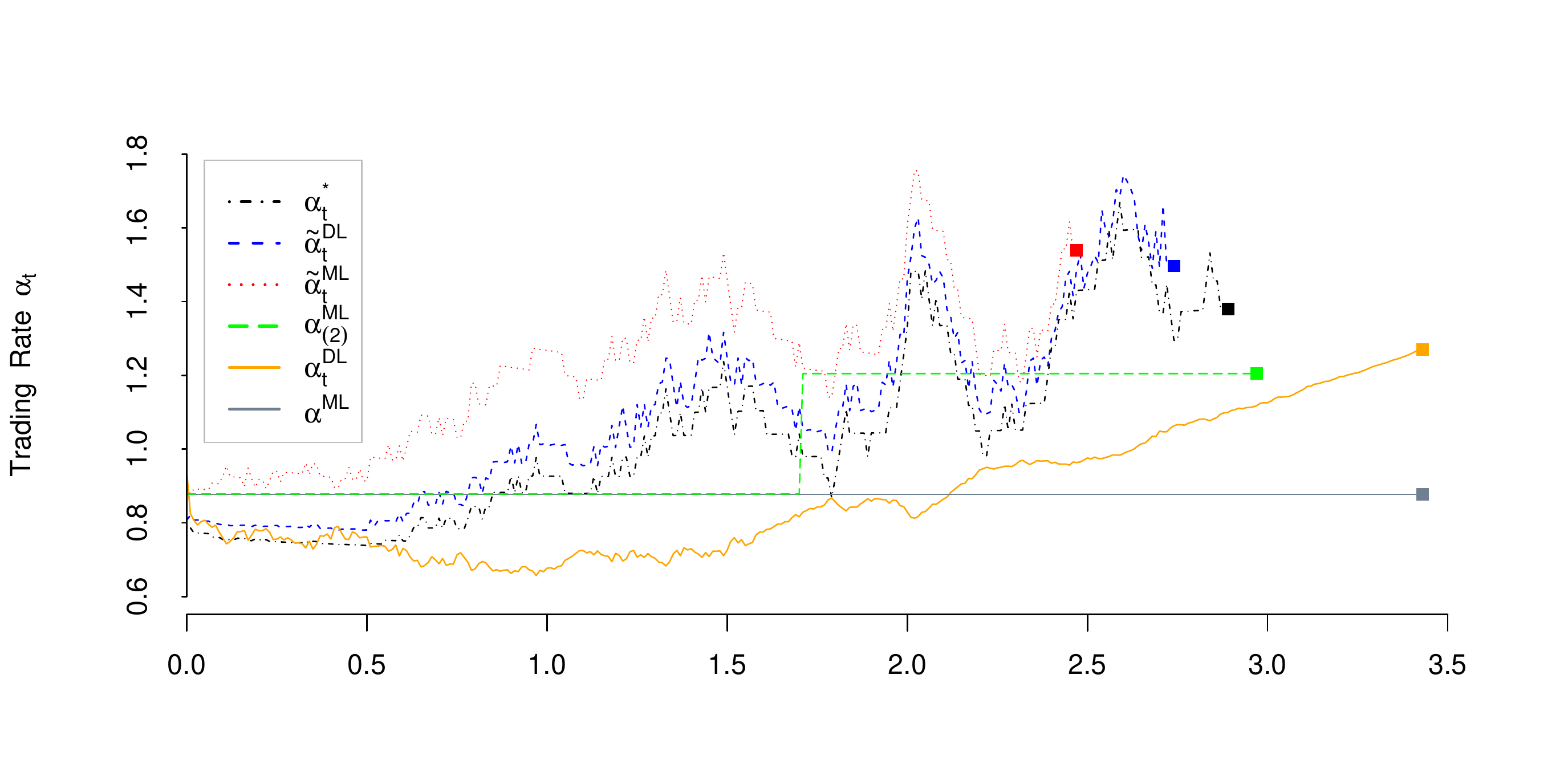}\\
\includegraphics[width=1\textwidth,height=1.5in, trim=0in 0.2in 0in 0.5in]{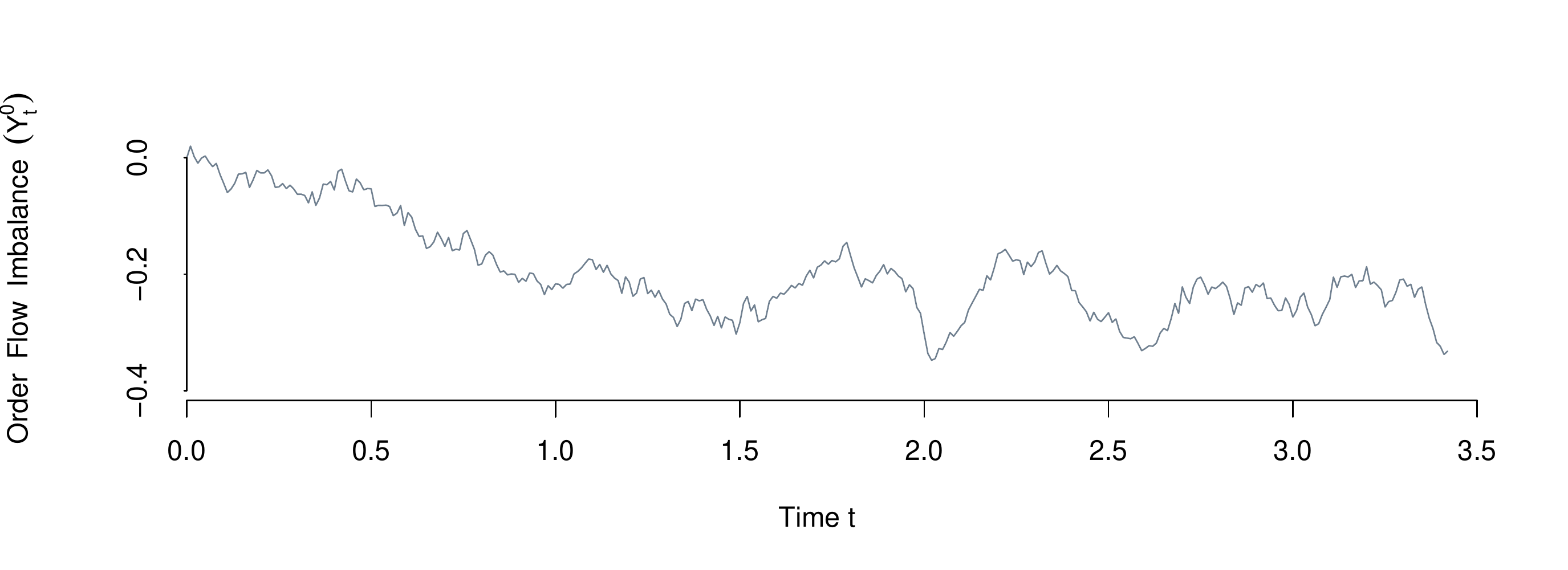}
\end{minipage}

\caption{Comparison of trading rates $(\alpha_t)$ for each of six strategies in Table \ref{tab1} given the shown simulated path of $(Y^0_t)$ (The realized $(Y_t)$ depends on the strategy chosen). Note that each strategy terminates at a different $T_0$ indicated with a square. \label{fig:rates-comp}  }
 \end{figure}

\subsection{Comparative Statics}
Focusing on a single strategy, $\talph^{DL}_t$, we briefly discuss how adjusting the values for parameters $c$, $\eta$, and $\kappa$ affect the trading rate and realized horizon.  Increasing $c$ corresponds to lower tolerance for timing risk and intuitively leads to a shorter realized horizon and an increase in trading rate across all values of $y$.  Choices for $\kappa$ and $\eta$ depend respectively on the trader's assessment of the added cost of transacting when order flow is unbalanced and exactly how susceptible one is to revealing information to other participants.  Increasing $\kappa$ raises sensitivity to the order flow imbalance which leads to an increase in trading speed and shortened execution horizon, particularly when $(Y_t^2)$ moves away from $0$.  At the other extreme, setting $\kappa=0$ leads to the strategies addressed in Lemma \ref{lem:exec-curve} with zero informational cost.  The effects of increasing $\eta$ depend on the market state, increasing the trading rate when order flow tilts towards buy orders and slowing when order flow is balanced or sell orders dominate.  Specifically, an increase in $\eta$ means a stronger trade impact on the order flow process, and thus it is beneficial in a buy market to tolerate somewhat higher instantaneous costs because the trader can more efficiently capture the savings that result from more balanced order flow in the future (and vice-versa in a sell-tilted market).  Figure \ref{fig:tratecomp} illustrates these comparative statics for $\talph^{DL}$ in terms of $\kappa$ and $\nu$.
Note that while theoretically $\alpha^{DL}$ from \eqref{eq:a-star3} could be negative, in all our plots $\alpha^{DL}$ remains far from zero and well-behaved.

    \begin{figure}[ht]
    \centering
   \includegraphics[height=2.5in,trim=0.1in 0.3in 0.1in 0.5in]{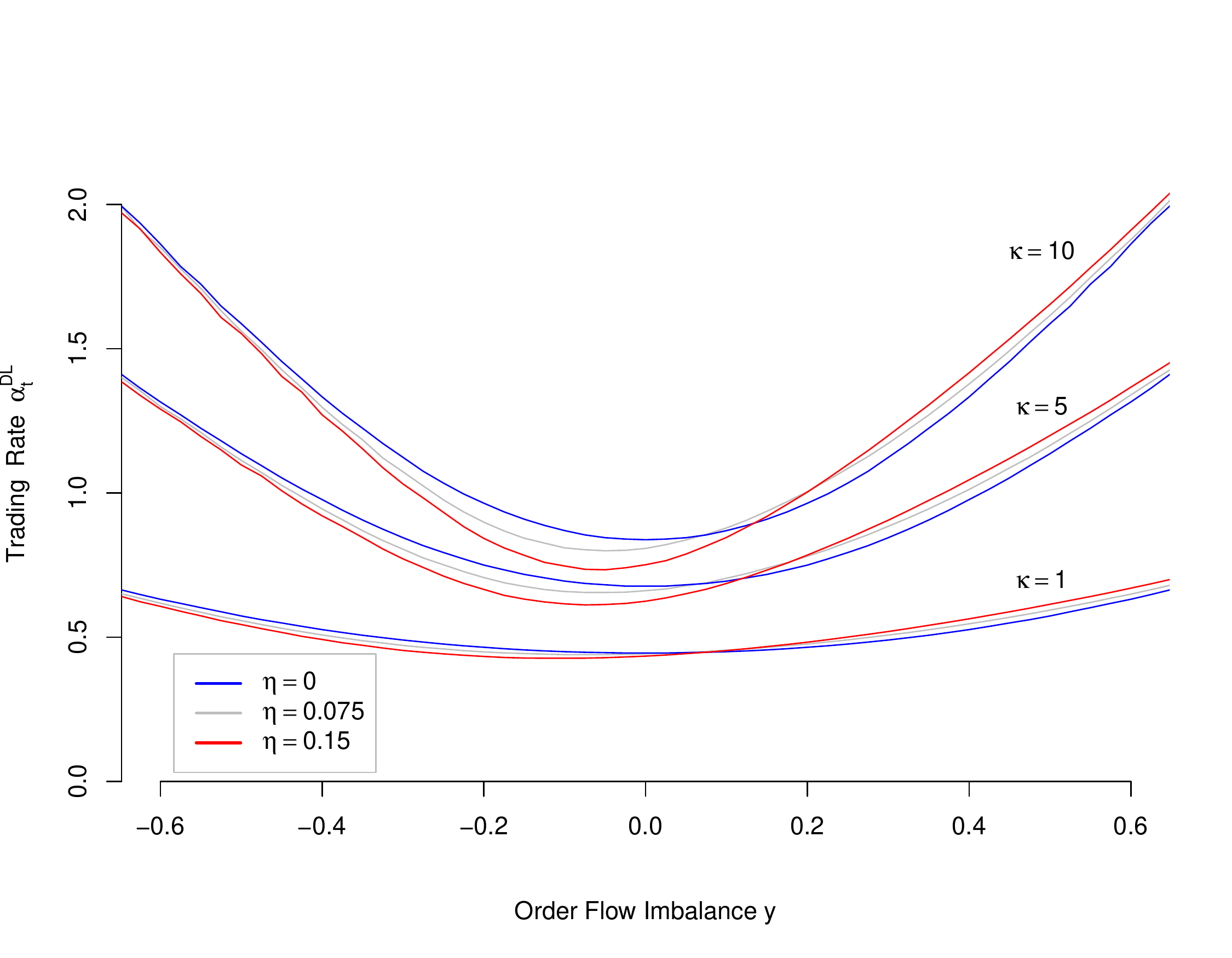}
   \caption{Trading rates $\talph^{DL}(x,y)$ plotted as a function of flow imbalance $y$ for different values of informational cost $\kappa$ and information leakage strength $\eta$. Inventory level is fixed at $x=3$.\label{fig:tratecomp}  }
    \end{figure}

\subsection{Realized Execution Horizon}
 We now explore some features of the realized execution horizon $T_0(x,y)$ when following the dynamic strategy $\talph^D_t(x_t,Y_t)$.  The left panel of Figure \ref{fig:fT-D} highlights the distribution of $T_0(x,y)$ for different initial market states. We observe that $T_0$ tends to be longest in a balanced market. Indeed, in that case the trader pays the most attention on minimizing his footprint and instantaneous execution costs, and therefore trades slowly. With positive imbalance $Y_t>0$, he is incentivized to trade at a more rapid pace in order to bring the market into a more balanced state.  On the other hand, when a market is dominated by sell orders $Y_t< 0$, the trader finds himself competing for liquidity and trading occurs at an even faster pace. The asymmetric effect of these effects creates a skew even with a symmetric informational cost $\kappa y^2$. This phenomenon is further shown in the right panel of Figure \ref{fig:fT-D} that shows a scatterplot of $T_0$ against terminal $Y_{T_0}$. It is also clear that the issue is not only whether order flow is balanced versus unbalanced, rather the side of the trade is very pertinent to the optimal strategy and realized horizon.  Generally, imbalanced order flow results in higher trading costs and a shorter horizon, but clearly as one would expect, trading against the prevailing order flow (selling when order flow is dominated by buy orders) is preferable.

We remark that with $T$ fixed, execution rate decreases with $Y_t$ in hopes that the order flow process will revert to a more balanced state and the cost of liquidity will decline in the future.  However, allowing the horizon to be adjusted brings a new incentive for accelerating execution in order to exit the market altogether and stop information leakage.  This is a phenomenon seen especially in times of panic or capitulation when minimizing the footprint is less important than finding liquidity, even at greater cost.

We also observe a strong correlation between realized execution cost and realized execution horizon.  Unbalanced order flow results in higher costs from the $Y_t^2$ term.  In addition, as lopsided order flow causes trading to accelerate, the trader also incurs higher costs from the instantaneous cost term $\alpha_t^2$.  So costs tend to be lower for longer realized execution horizon. Figure \ref{fig:realized-oi} provides another perspective on this feature by highlighting several specific inventory trajectories and the corresponding realized order flow paths. It also shows that the spread in realized horizon $T_0$ is significant and can be up to 50\% of the static $T^*$.

 \begin{figure}[ht]
 \centering
 \begin{tabular}{cc}
 \includegraphics[height=2.5in,trim=0in 0.2in 0in 0.3in]{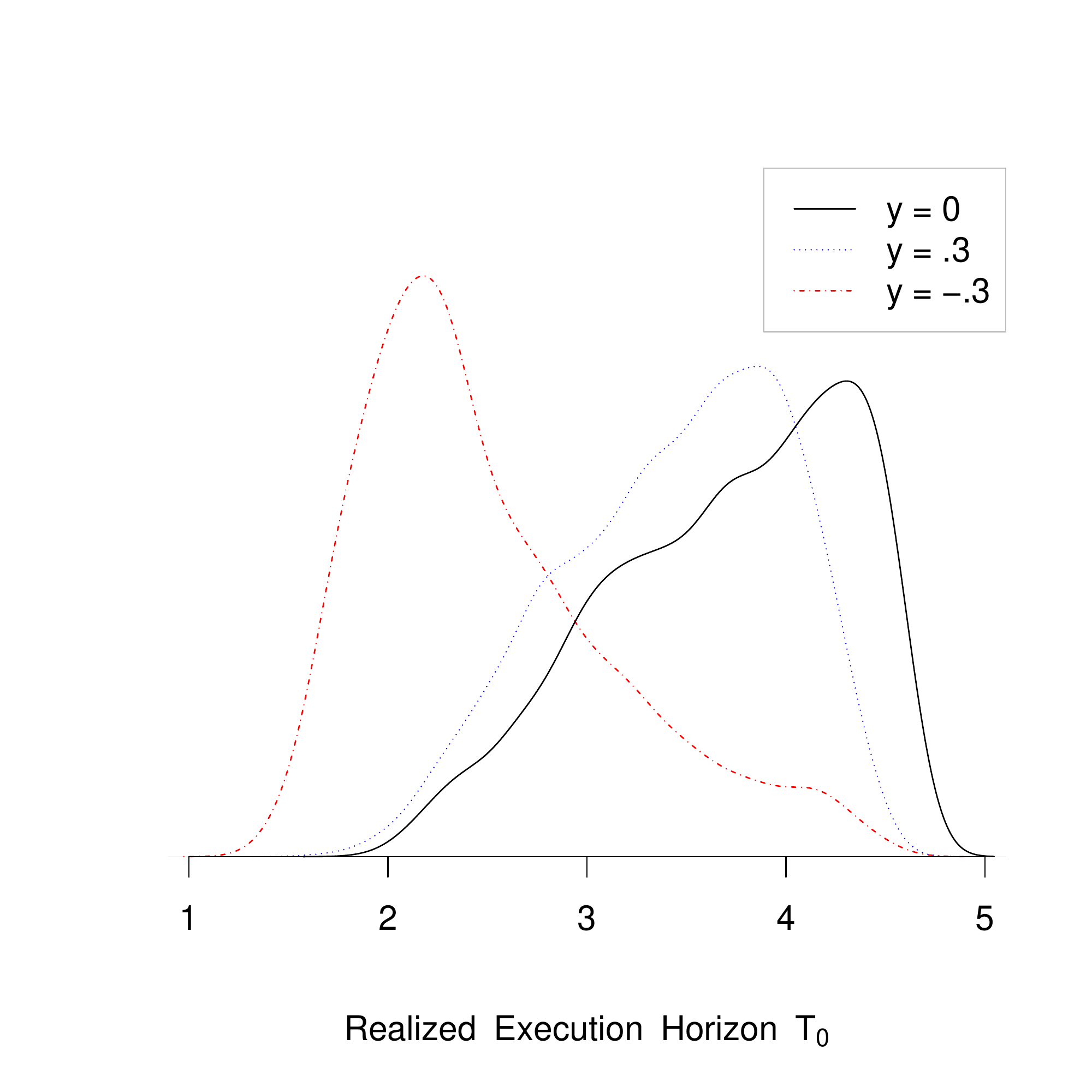} & 
 \includegraphics[height=2.5in,trim=0in 0.2in 0in 0.3in]{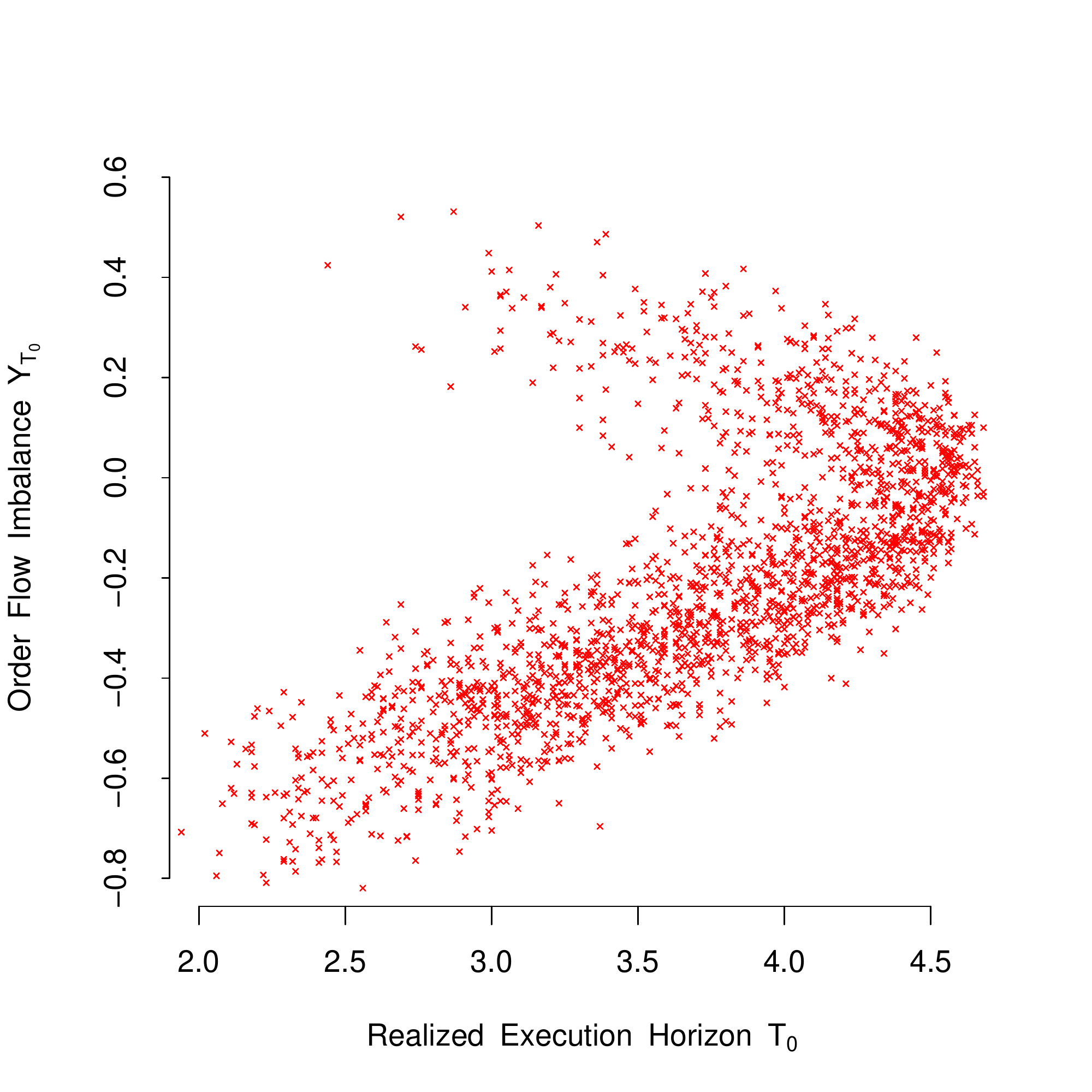}
  \end{tabular}
 \caption{Left: Distribution of realized execution horizon $T_0$ following strategy $\talph^{DL}$ for different values of initial flow imbalance $Y_0=y$. Right: Realized execution horizon $T_0$ against final order flow imbalance $Y_{T_0}$ when initial imbalance $Y_0=0$. 
 \label{fig:fT-D}}
 \end{figure}


\begin{figure}[ht]
\centering
\includegraphics[height=2.2in,width=0.9\textwidth,trim=0in 0.25in 0in 0.1in]{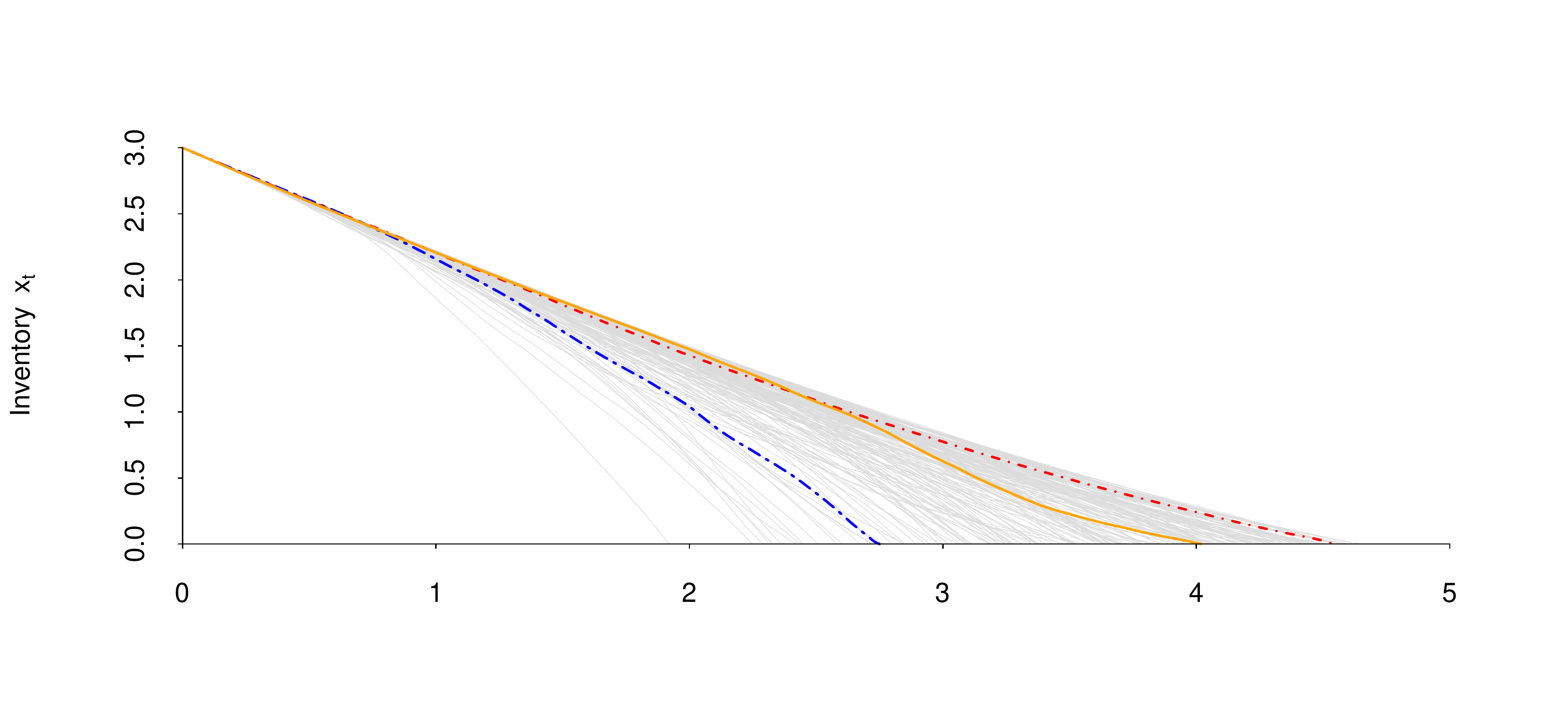}
\includegraphics[height=1.5in,width=0.9\textwidth,trim=0in 0.25in 0in 0.5in]{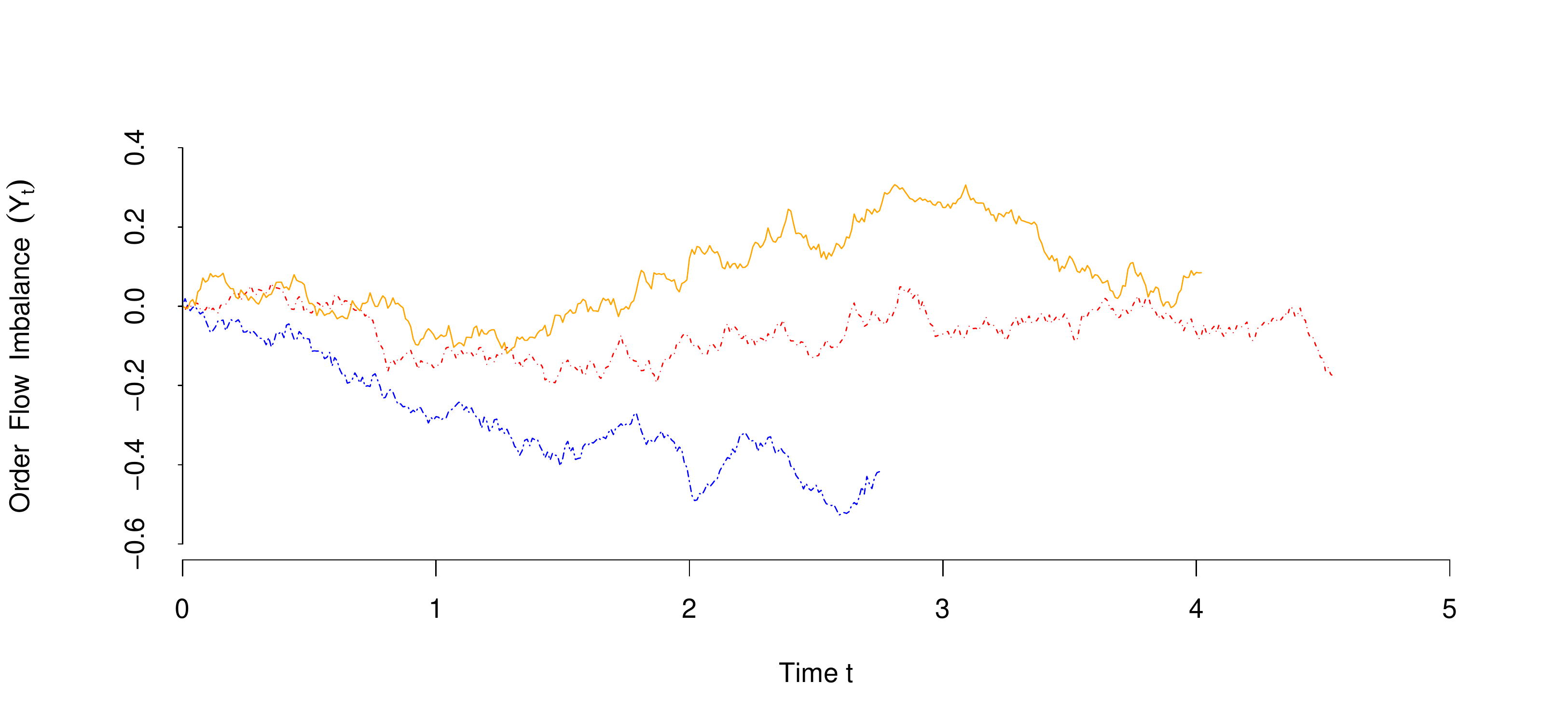}
\caption{Top: 200 simulated trajectories $(x_t)$ from dynamic adaptive strategy $\talph_t^{DL}$.  Highlighted are three trajectories resulting from different realized order flow imbalance  $(Y_t)$ paths.  Bottom: Corresponding realizations of order flow imbalance $t\mapsto Y_t$. \label{fig:realized-oi}}
\end{figure}

\subsection{Static Information Leakage}\label{sec:static}
One of the motivations for the present analysis was the work of Easley et al.~\cite{easley2012optimal} (ELO), who considered a related static optimal execution horizon model. Through the lens of our setup, \cite{easley2012optimal} treated the case where informational costs are measured only through $Y_{T}$ rather than through the integral term in \eqref{oppr}. Specifically, ELO equated informational footprint to the absolute value of the terminal flow imbalance $|Y_T|$. Also, ELO (implicitly) assumed a VWAP execution strategy on $[0,T]$ which is equivalent to taking zero inventory risk $\lambda(x) = 0$ and instantaneous impact $\int_0^T \dot{x}_s \,ds$ and translates to the myopic strategy $\alpha^{ML}$ of constant trading rate. Finally, timing risk was modelled directly as $\Lambda(T) = c \sqrt{T}$, motivated by the same structural form for volatility of $S_T$.  The overall problem in \cite{easley2012optimal} was therefore
\begin{equation} \label{eq:Eas}
\min_{T \ge 0} \left\{\mathbb{E} [ |Y^\alpha_T|] +c \sqrt{T} \right\}, \qquad \alpha_t = x/T.
\end{equation}
Our framework allows treatment of \eqref{eq:Eas} in a dynamic setup, i.e.~beyond the myopic strategies and beyond a static optimization to obtain $T^*(x,y)$.

The possibility utilized in \cite{easley2012optimal} to directly incorporate a timing cost of the form $\Lambda(T)$ can be easily handled in more generality within \eqref{oppr} since the latter term makes no difference to the fixed-horizon problems in Sections \ref{sec:myopic}-\ref{sec:linear-quadratic} and hence only shows up in the second-step optimization over $T$. Based on numerical experiments, replacing running $Y$-costs with a terminal cost $\propto Y^2_{T_0}$, tends to slow the optimal trading strategy in sell dominated markets which allows the mean reversion in the order flow process  to kick in and lower the terminal cost.  In the presence of positive order flow, the change in trading rate can be in either direction depending on $\phi(\alpha)$ and the trade-off between instantaneous costs and informational costs.


\section{Model Calibration}\label{sec:implement}
To implement the proposed execution strategies, the trader must observe the order flow imbalance $Y_t$. Moreover, they need to be able to calibrate the parameters of $Y_t$. This requirement is different from typical execution strategies that operate in ``open-loop'' settings, i.e.~without any immediate input of market data. Of course, most empirical trading is ``closed-loop'' and dynamically responds to market messages. In this section we briefly discuss such calibration and translation of market information into model inputs.

\subsection{Empirical Order Flow}\label{sec:empirical-flow}

We begin by focusing on \emph{executed} orders, which from the flow point of view can be summarized as a sequence $V_1, V_2, \ldots,$ where $V_k$ is the signed market order volume (positive for buys and negative for sells) for the security in question. We assume that trading is in volume time, so there is no separate time-stamp component. A raw order flow would then be the cumulative sum $\sum_k V_k$. Assuming that the participants focus on recent trades (i.e.~market memory is limited) leads to consideration of moving averages of $V_k$.  By analogy to the discussed Ornstein-Uhlenbeck dynamics we therefore introduce the following exponentially weighted moving average (EWMA) flow imbalance process $(I_k)$ that is defined recursively via
\begin{align}\label{eq:ewma}
I_{k+1} = e^{-\beta|V_k|} I_k +(1-e^{-\beta|V_k|})\sgn(V_k),
\end{align}
where the memory parameter $\beta$ is a proxy for the time-scale of market participants persistency of beliefs about order flow.
Intuitively, if all trades were of unit volume, we would have $I_{k+1}=e^{-\beta} I_k + (1-e^{-\beta}\sgn(V_k)$; treating a single trade of $|V_k|$ as that many unit-volume trades leads to \eqref{eq:ewma}. We suggest that $\beta = a/V_{daily}$ where $V_{daily}$ is the average daily volume and $a \in [10,100]$ is the intra-day mesoscopic time-scale of order flow. By construction, $I_k$ takes values in $[-1,1]$, with $I_k=0$ representing a balanced market, and positive and negative values of $I_k$ representing a market tilted towards buying and selling respectively. Figure \ref{fig:empoi} shows a typical daily path of $I_k$ for two different values of $\beta$. To draw the Figure we considered all executed Nasdaq ITCH trades between 9:40am and 3:55pm on a fixed trading day and initialized $\tilde{I}_0  =I_0 =0$ in the beginning (note that if one fully adheres to the concept of moving averages, $I_0$ should include flow from the previous day, but this is rather problematic to properly implement). As expected, we observe a mean-reverting behavior due to strong negative auto-correlation in $V_k$'s. The value of $\beta$ controls the volatility of $I_k$.

Moving beyond \eqref{eq:ewma} requires separate investigation beyond the scope of this paper. Indeed, the modern limit order book messages are much richer than just a sequence of signed volumes. There are issues of inter-trade durations, intra-day seasonality, and other market information. Moreover, the data features can very significantly between more liquid and less liquid names.
Table \ref{tab2} summarizes some basic details about order flow within a given trading week across three representative equity tickers. We note that each day one must process data for tens of thousands of executed trades, but in fact this is a tip of the iceberg, since there are hundreds of thousands of other trades recorded in the LOB. 

\begin{figure}[ht]
\centering
\includegraphics[height=2.5in,trim=0.15in 0.25in 0.15in 0.25in]{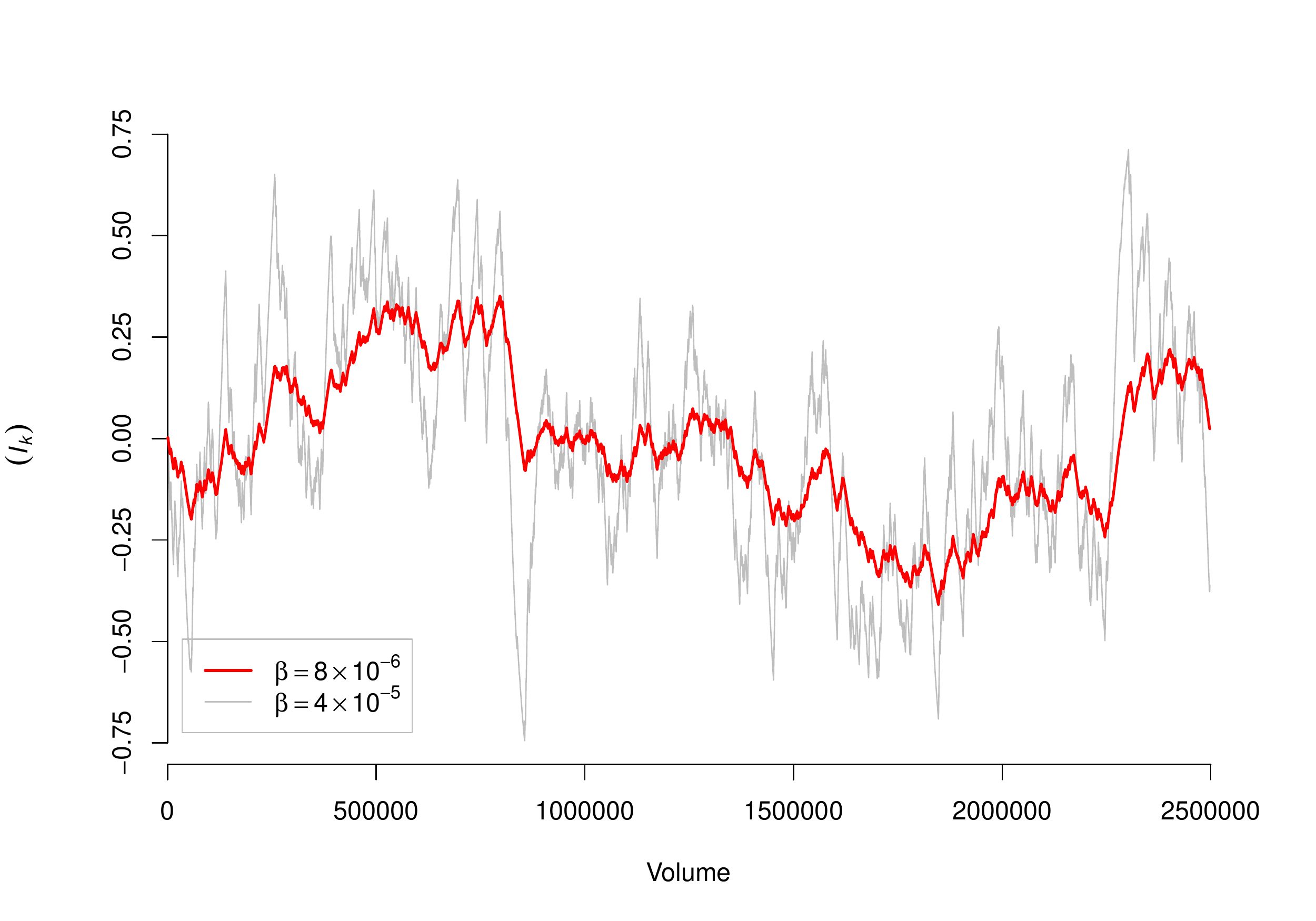}
\caption{The EWMA order flow imbalance metric for Teva Pharmaceutical (ticker: TEVA) for a single day 5/3/2011.  The data includes executed orders from NASDAQ, BATS and Direct Edge exchanges which accounted for $2,497,623$ of the $8,059,668$ total traded shares on the day.  We also show the VPIN-like metric that used $V=25,000$ and $n=20$ in \eqref{eq:realized-imbalance} and \eqref{eq:VPIN} respectively. \label{fig:empoi}}
\end{figure}

\begin{remark}
While traditionally only executed trades were used to define order flow, in limit order books this approach is questionable. Indeed,
in the 2010s, market orders represent less than 5\% of total order flow, see Table \ref{tab2}. At the same time, as far as the snapshots of the LOB queues and respective depths are concerned, limit orders that add/cancel at the touch (highest bid/lowest ask price levels) are effectively equivalent to market orders. Therefore, one probably ought to include such limit orders when computing $I_k$ in \eqref{eq:ewma}. Table \ref{tab2} shows that the distribution of limit orders is markedly different from market orders, with limit orders tending to be about 50\% smaller on average. Limit orders deeper into the books may also carry material information and could be considered (with a respective weight depending on the corresponding depth level). The precise way to define order flow is therefore ambiguous.
\end{remark}

\begin{table}[ht]
\footnotesize
\begin{center}$$
\begin{array}{rrrrr}
\hline
\multicolumn{5}{ c }{ \text{Daily Order Summary}} \\
& & \text{MSFT} & \text{TEVA} & \text{BBBY} \\ \hline \hline
& \text{Volume} &32,748,575&3,809,972&1,222,164\\
\text{Executions} &\text{$\#$ of Orders} &26,814&13,811&8,302 \\
&\text{Avg trade size} &1,222&271&148 \\  \hline
\text{Limit}& \text{Volume}&271,383,713&25,523,228&13,558,498  \\
\text{Orders}& \text{$\#$ of Orders}&577,164&182,541&127,362 \\
\text{at Touch}& \text{Avg order size}&469&139&107 \\ \hline
& \text{Volume}&675,466,734&98,398,817&70,149,547 \\
\text{Total}&\text{$\#$ of Orders} &1,284,524&542,762&534,557 \\
&\text{Avg order size} &526&181&132 \\  \hline
\end{array}$$
\end{center}
\caption{Summary daily statistics for LOB orders on the tickers MSFT (Microsoft), TEVA (Teva Pharmaceutical) and BBBY (Blackberry) from NASDAQ, BATS BZX, EDGA and EDGX exchanges. We report averages for the week (5 trading days) of 5/2/2011-5/6/2011.
\label{tab2}}
\end{table}

An alternative way to define empirical order flow imbalance is based on bucketing. This approach is more in line with a discrete model, such as the one in ELO \cite{easley2012optimal} and indexes flow by equally-sized volume slices rather than by individual trades.
Namely, consider (executed) volume slices of size $V = V^{B}_\ell+V^{S}_\ell$ where $(V^{B}_\ell)$ and $(V^{S}_\ell)$ represent the buy and sell volume respectively for the $\ell$-th slice. The bucket flow imbalance $\tilde{I}_\ell$ is
\begin{equation}\label{eq:realized-imbalance}
\tilde{I}_\ell :=\frac{V^{B}_\ell-V^{S}_\ell}{V}=2V^{B}_\ell-1.\end{equation}
Compared to \eqref{eq:ewma}, we have the link $V \simeq 1/\beta$ to achieve same time-scale for $(I_k)$ and $(\tilde{I}_\ell)$.
The difficulty with \eqref{eq:realized-imbalance} is that order volumes are heavy-tailed and a large order that spills over multiple volume buckets must be somehow handled. Moreover, $\tilde{I}$ in \eqref{eq:realized-imbalance} is intrinsically tied to the chosen bucket size $V$ (See \cite{andersen2012vpin} and \cite{easley2012vpin} for different arguments regarding appropriate bucket sizes). Finally, as in \eqref{eq:ewma}, a decision must be made how to include limit order messages when defining volume buckets and respective $V^B$ and $V^S$.

The defined $I_k$ and $\tilde{I}_\ell$ are directly observed, and one could attempt to use them as the basis
for the flow process $(Y_t)$ used in the previous sections. According to \cite{easley2011microstructure,easley2012flow,easley2012optimal,easley2012vpin}, informational costs arise from order flow \emph{toxicity} which is in turn tied to the participants' beliefs about probability of adverse selection. Hence, translating {past} information contained in $I$ into $Y$ requires making a judgement on how such beliefs about future order flows are formed. It is commonly accepted  that certain trades are influential or informative while others have little to no impact on the market.  With this in mind, it follows that a trade may have little influence on a market maker's expected flow imbalance even if the associated trade volume $V_k$ was large. Thus, the private information leaked to the market by a trade is the product of numerous factors beyond trade size: spacing of successive orders, prevailing market state, LOB shape, etc. Consequently, the exact relation between $(I_k)$ and $(Y_t)$ remains open. Further questions about the most relevant time scale or how the recent history of observed order imbalance might influence the expected future order imbalance
 are postponed to future empirical study.

\begin{remark}
In \cite{easley2012vpin}, ELO contend that order flow toxicity is linked to the level of ``informed'' trading and can be approximated via the following VPIN metric based on \eqref{eq:realized-imbalance},
\begin{align}\label{eq:VPIN}
\text{VPIN}_\ell=\frac{1}{n} \sum_{k=\ell-n}^{\ell-1} |\tilde{I}_k|,
\end{align}
where $n$, chosen along with bucket size $V$ represents the window relevant for persistency of order flow. Thus, VPIN  is a moving average of observed values of $\tilde{I}_\ell$ for $n$ latest volume slices. According to ELO, VPIN is a good approximation to the market-makers' expectations of the current imbalance $\text{VPIN}_\ell \approx {\E| \tilde{I}_{\ell} |}$ and hence can be used as a proxy for $Y_t$ in \eqref{eq:Eas}. In other words, the traders should act to minimize their impact on VPIN, which is the expected absolute order flow imbalance. Because VPIN is directly based on observed traded volumes, this also allows an explicit definition of trader's informational impact, see \eqref{eq:phi-I} below. ELO suggest to take $V = V_{daily}/50$ and $n=50$ which makes VPIN a daily moving average of flow imbalance. This seems rather long and in Figure \ref{fig:empoi} we make $n$ smaller to focus on intra-day scale.
\end{remark}

\subsection{Discrete Time Formulation}\label{sec:discrete-time}

Building on \eqref{eq:realized-imbalance} and \eqref{eq:VPIN} one can construct a discrete model for optimal execution. This involves
reinterpreting the strategy $\alpha$ as a participation rate based on the observation that
 an executed sell trade inherently affects the next bucket imbalance $\tilde{I}_\ell$ since it physically displaces some of the other volume from that bucket. A discrete-time model also allows to examine more general costs and model dynamics.

Fixing a volume bucket $V$, we assume that the trader chooses a participation rate $\alpha_k$ at each step, where $$x_{k+1}=x_k-\alpha_k V$$ and $k$ indexes trade volume. The trader's participation influences the flow imbalance at the next step via
\begin{equation}\label{eq:affected-I}
Y_{k+1}=F(Y_k,\eps_{k+1})-\phi(\alpha_k, Y_k)
\end{equation}
where $\eps_{k+1}$ are independent random perturbations, $F(y,\cdot)$ models the dynamics in $Y_k$ that happen apart from the trader and $\phi(\alpha,y)$ is the information leakage given previous imbalance $Y_k=y$. One motivation to generalize the leakage function is the trade influence proposed in \cite{easley2012optimal},
\begin{equation}\label{eq:phi-I}
Y_{k+1}=\psi(\alpha_k)({Y}_{k}(1-\alpha_k)-\alpha_k)+(1-\psi(\alpha_k))Y_{k} + \eps_{k+1},\end{equation}
where we still have $\alpha_k \in [0,1]$ and the function $\psi\in [0,1]$ is monotonic increasing.  The new expected order flow imbalance is a convex combination of two extreme outcomes: full leakage according to $\psi$ (first term) and no leakage (second term).  \eqref{eq:phi-I} simplifies to
\begin{equation*} 
\E[Y_{k+1}]=Y_k-\alpha_k \psi(\alpha_k)(Y_k+1)=:Y_k-\phi(\alpha_k, Y_k).
\end{equation*}

In analogue to Section \ref{sec:model}, the trader's goal is to minimize total expected costs until the entire position has been liquidated,
\begin{align}\label{eq:discrete-v}
v(x,y) :=\inf_{\alpha} \E_{x,y} \left[ \sum_{k=0}^{T_0-1} g(\alpha_k) + \kappa Y_k^2 + \lambda(x_k)\right ], \quad T_0 = \min \{ k: x_k = 0\}.
\end{align}
The control $\alpha_k$ is constrained so that $\alpha_k \in (0,1]$ and is assumed to be in feedback form, $\alpha_k = \alpha(x_k, Y_k)$.


The indefinite-horizon control problem \eqref{eq:discrete-v} can be solved by introducing an auxiliary ``time'' variable $t$ such that execution stops after $t$ steps (if it did not terminate already) and remaining inventory at $t$ incurs a terminal cost $H(x_t)= A x_t^2$ (assuming immediate one-step liquidation after $t$, cf.~the VWAP strategy in \eqref{eq:x-star-lin}). This auxiliary problem has value function $v^{(t)}$ defined in \eqref{eq:discrete-dp}. As $t \to\infty$, this execution horizon constraint vanishes and we expect to recover the time-stationary solution $v(x,y)$ of \eqref{eq:discrete-v}.

Using $t$ as time-to-maturity, we have the discrete-time dynamic programming equations
\begin{align}v^{(0)}(x,y) &=H(x),\notag \\
\label{eq:discrete-dp}
v^{(t)}(x,y) &=\inf\limits_{\alpha \in (0,1]} \mathbb{E}_{y}\Bigl[ g(\alpha) + \kappa y^2 + \lambda(x) +v^{(t-1)}(x-\alpha V, Y_1)\Bigr]\end{align}
where  $Y_1 = Y_1^\alpha$ is defined by \eqref{eq:affected-I} and $v^{(t)}(0,y) =0 \forall y$. The one-stage
problem in \eqref{eq:discrete-dp} can be readily solved using a Markov-chain approximation method by discretizing the state space of $Y$ and the bounded control space of $\alpha$ (which also makes the state space of $x_t$ discrete).
%
%
The above procedure allows arbitrary dynamics for $(Y_k)$ beyond \eqref{eq:phi-I} since one can always use Monte Carlo or other methods to compute the transition density $p_{Y_1}( \cdot | Y_0=y,\alpha)$. 

\subsection{Future Directions}\label{sec:conclusion}
An important aspect that is missing from the presented models is \emph{price risk}. With a fixed horizon, the assumption that the unperturbed asset price is a martingale makes realized revenue only depend on execution risk. However, once the agent has the liberty to extend the execution horizon, this is no longer the case and the trader could also chase higher revenues. Explicit modeling of such objectives would necessarily increase the dimensionality since another stochastic state variable, the (mid-)price $S_t$, must be added. Nevertheless, under certain assumptions this more general setup could still be tractable. In particular, one could maintain the linear-quadratic structure by adopting the assumptions of \cite{Gatheral2011}. Gatheral and Schied \cite{Gatheral2011}  assume that $S$ is described by a geometric Brownian motion and  inventory risk is measured by time-averaged value-at-risk (VAR), i.e.~$\lambda_t = \lambda x_t S_t$. Making this adjustment in \eqref{oppr} leads to the problem of minimizing
\begin{align} \label{newproblem}
\check{u}(T,S,x,y) := \E_{S,x,y} \left[\int_{0}^{T} \left(\dot{x}^2_t + \kappa Y^2_t + \lambda x_t S_t \right) \, dt\right].
\end{align}
Assuming linear information leakage and that $S$ and $Y$ have  correlation $\rho \in [-1,1]$, $dW_t^{(S)} dW_t^{(Y)}= \rho dt$ allows for a closed-form solution
$$\check{u}(T,S,x,y) = x^2 \check{A}(T)+y^2 \check{B}(T)+x y \check{C}(T)+ S^2 \check{D}(T) +  S x \check{E}(T) + S y \check{F}(T) + \check{G}(T)$$
where the coefficients $\check{A},\ldots,$ solve yet again a Riccati ODE. Moreover, a fully closed form solution is possible for strategies that are myopic with respect to $Y_t$ (see \cite{gatheral2011optimal} for details). As in Remark \ref{rem:linear-riccati}, the main difficulty is that an inventory risk that is linear in $x$ cannot guarantee $x_s \ge 0$ and hence is likely to include buying, making the optimization over $T$ delicate.

In the dynamic case the execution strategy is given by
$$
\check{\alpha}^D(t,S_t,x_t,y_t) =\frac{1}{2} \Big\{x_t(2 \check{A}(T-t)+\eta \check{C}(T-t))+Y_t(\check{C}(T-t)+2\eta \check{B}(T-t))+S_t (\check{E}(T-t)+\eta \check{F}(T-t))\Big\}
$$
and is therefore linear in  price $S_t$. (The myopic strategy is also linear in $S_t$). Note that only the execution cost $\check{u}$ is impacted by $\rho$, while the strategies themselves are independent of the correlation between asset prices and order flow.
%
%
%
In reality, the joint behavior of order flow and asset prices remains poorly understood.  In fact, the positive correlation between market-order flow and asset price is a direct indication of adverse selection affecting liquidity providers, see e.g.~the very recent preprint \cite{CarmonaWebster14}. Further investigation into how the co-movement of order flow and asset prices might affect execution costs seems warranted.


Another piece that could be potentially added is LOB resiliency, i.e.~the intermediate effect of consuming liquidity which temporarily moves the execution price. As already explained, it is analytically quite similar to the modeled impact on $Y$, but instead operates directly on $S$. Resiliency allows to separate the temporary impact on the LOB that does not increase toxicity with the ``permanent'' impact on $Y$ that changes participant expectations of future flow and generates informational footprint.

\vspace{0.5in}
\appendix

\section{Finite Difference Approach to \eqref{eq:hjb2}}\label{App:AppendixA1}
We employ an explicit finite difference scheme to find an approximate solution to \eqref{eq:hjb2}.  Denote
$v_{i,j}:=v(x_i, y_j)$, where $x_i=i\Delta x$ for $i = 0, 1,\ldots, N$ and $y_j= y_0 + j \Delta y$ for $j=0,1, \ldots ,M$.
Derivatives of $v$ are approximated as
\begin{align*}
\frac{\partial}{\partial x } v(x_i, y_j)&=\frac{v_{i+1,j}-v_{i,j}}    {\Delta x};\\ \notag
\frac{\partial}{\partial y } v(x_i, y_j)&=\frac{v_{i,j+1}-v_{i,j-1}} {2\Delta y};\\ \notag
\frac{\partial^2}{\partial y^2 } v(x_i, y_j) &=\frac{v_{i,j+1}-2v_{i,j}+v_{i,j-1}}{(\Delta y)^2},
\end{align*}
and we apply the boundary condition $v(0,y)=v_{0,j}=0 \forall j$.   In $y$ we use the boundary conditions for $v_{i,0}$ and $v_{i,M}$ via  $\frac{\partial^2 v(x_{i},y_1)}{\partial y^2 }=\frac{\partial^2 v(x_{i},y_{M-1})}{\partial y^2 }=0$ and choose $y_0$ and $y_M$ such that $\mathbb{P}((Y_t)\notin [y_0,y_M])\approx 0$.  Substituting into \eqref{eq:hjb2} we have
\begin{align*}
0 =  \frac{1}{2} \sigma^2 \frac{\partial^2 v_{i,j}}{\partial y^2 }-\beta y_j \frac{\partial v_{i,j}}{\partial y }  + \kappa y^2_j + \lambda(x_i)-\frac{1}{4} \left(\frac{\partial v_{i,j}}{\partial x }+\eta \frac{\partial v_{i,j}}{\partial y }\right)^2
\end{align*}
and rearranging terms yields
\begin{multline*}
v_{i+1,j} =v_{i,j}  \\ + \Delta x \left(2\left(\kappa y^2_j+\lambda(x_i)-\beta y_j \frac{v_{i,j+1}-v_{i,j-1}} {2\Delta y}+\frac{\sigma^2}{2} \frac{v_{i,j+1}-2v_{i,j}+v_{i,j-1}}{(\Delta y)^2}\right)^{1/2} -\eta \frac{v_{i,j+1}-v_{i,j-1}} {2\Delta y}\right).
\end{multline*}

So we have a so called ``time-marching" scheme in $x$, where $v$ at each inventory level $i+1$ can be approximated by values from the previous inventory level $i$.  This explicit approach is available for appropriate parameter values in which $\alpha^*>0$ holds.  For more extreme parameter values we may have $\alpha^*=0$ (recall that we constrain $\alpha^*\geq 0$).  In this case our explicit method would fail for certain grid points and it would be necessary to utilize an alternative method.  

\section{Proof of Lemma \ref{lemma1}} \label{App:AppendixA}
\begin{proof}
We skip the trivial case $\lambda(x)=0$.  See \cite{grinoldactive1999} for details when $\lambda(x)=cx^2$.  For $\lambda(x)=cx$, when $x$ is unconstrained, the problem is a straightforward application of the Euler-Lagrange equation.  For cost functional $F(x,\dot{x},t)= \dot{x}_s^2 + c x_s$, the optimal trajectory $x^{MQ}_t$ must satisfy
\begin{align}
\left(\frac{dF}{dx}- \frac{d}{dt}\frac{dF}{ d\dot{x}}\right)=0
\end{align}
Applying the constraint that $x$ is decreasing introduces the boundary condition at $x=0$.  Then the optimal trajectory for the constrained minimization problem either lies on the curve which satisfies the Euler-Lagrange equation or lies along the boundary, with the transition from the former to the latter occurring where $x^{MQ}_t$ is tangent to the line $x=0$.  In order to find the point at which this transition takes place, we find $T'$ such that $\frac{dx^{*}_t}{dt}{\bigr |}_{t=T'}=0$.  It is easily verified that  $T'=\frac{2\sqrt{x}}{\sqrt{c}}$ satisfies the requirement and so we have $\hat{T} = \min(T, \frac{2\sqrt{x}}{\sqrt{c}})$.

Having computed $x^{M}_t$ for each choice inventory risk term, $\alpha^{M}_t$ and $\cI^{M}$ are computed respectively by differentiating with respect to $t$ and integrating over the interval $[0,T]$ ($[0,\hat{T}]$ for $\cI^{MQ}$).
 \hfill\end{proof}

\section{Proof of Proposition \ref{lemma1}} \label{App:AppendixB}
 \begin{proof}
An application of Fubini's theorem, permits the interchange of expectation and integration leaving us with a straight-forward but lengthy integral computation. The general expression for $\cO$ is
\begin{align}\label{eq:A-int}
\cO(T,x,y) = \int_0^T (y e^{-\beta t} - A_t)^2 + \frac{\sigma^2}{2\beta} (1-e^{-2\beta t}) dt
\end{align}
where
\begin{align}\label{eq:def-A}
A_t := \int_0^t e^{-\beta(t-s)}\phi_s ds
\end{align}
 captures all the information leakage.
Integrating the other two terms gives
\begin{align*}
\cO^{0}(T,x,y) & = \kappa \int_0^T (y e^{-\beta t})^2 + \sigma_t^2 dt \\
 & =  \frac{\kappa y^2}{2\beta} \left(1-e^{-2\beta T}\right) + \frac{\kappa \sigma^2}{4\beta^2}\left(
  2\beta T+e^{-2\beta T}-1\right).
\end{align*}
Computing the integral for the remaining three cases is straight-forward but tedious.  We provide the values of $A_t$ for each case.
For $\phi_t=\eta \alpha_t^{ML}$, we have $A^{ML}_t= \frac{\eta x}{\beta T}(1-e^{-\beta t})$.
 When $\phi_t=\eta \alpha_t^{MQ}$ we have
 \begin{align*}
 A^{MQ}_t &=  \eta {\int_{0}^{t}\left(\frac{c \hat{T}}{4}+\frac{x}{\hat{T}}-\frac{c t}{2}\right) e^{-\beta(t-s)} \, ds}\\
 &=\frac{\eta}{\beta^2}\left(c(1-\beta t+e^{-\beta t})+\beta (1-e^{-\beta t})\left(\frac{c \hat{T}}{4}+\frac{x}{\hat{T}}\right)\right).
 \end{align*}
 Lastly, when $\phi_t=\eta \alpha_t^{MH}$ and $\sqc\neq\beta$, we have

 \begin{align*}
 A^{MH}_t &= \eta {\int_{0}^{t}\frac{ \sqc x \cosh( \sqc(T-s) )}{\sinh(\sqc T)} e^{-\beta(t-s)} \hspace{1mm} ds}\\
 &=\frac{\eta \sqc x}{(c-\beta^2)\sinh(\sqc T)} \large{(} \lambda\sinh(\sqc(T-t)) +\sqc e^{-\beta t} \sinh(\sqc T)\\
 & \qquad -\beta \cosh(\sqc (T-t))+\beta e^{-\beta t} \cosh(\sqc T) \large{)}.
 \end{align*}

 If $\sqc=\beta$ then the final expression simplifies to
 \begin{align*}
 A^{MH'}_t=\frac{\eta x e^{-\beta t} e^{-\beta T} \left(2\beta t e^{2\beta T}+e^{2\beta t}-1\right)}{4\beta \sinh( \beta T)}.
 \end{align*}

One can finally integrate (using a symbolic integration software for example) \eqref{eq:A-int} over $t \in [0,T]$ to obtain closed-form expressions for $\cO^{MQ}$ and $\cO^{MH}$. The resulting formulas which take several lines are available on request from the authors.
\hfill \end{proof}

 \section{Proof of Proposition 2} \label{App:AppendixC}
\begin{proof}
 Substituting \eqref{eq:riccati} into the HJB PDE \eqref{eq:fpde2} we have
 \begin{align*}
 u^D_T &=\sigma^2 B(T)+\kappa y^2+c^2 x^2-\beta y(2yB(T)+xC(T))-
 \left( \frac{2xA(T)+yC(T)+\eta(2yB(T)+xC(T))}{2}\right)^2\\
 &=x^2 \left(c^2-(A(T))^2- \eta A(T)C(T)-(\eta C(T))^2\right)\\
 &+y^2\left(\kappa-2\beta B(T)-\left(\frac{C(T)}{2}\right)^2-\eta B(T)C(T)-(\eta B(T))^2\right)\\
 &+xy\left(-\beta C(T)-A(T)C(T)-2\eta A(T)B(T)-\frac{\eta (C(T))^2}{2}-\eta^2 B(T)C(T)\right)+\sigma^2 B(T).
\end{align*}
On the other hand, $u^D_T = x^2 A'(T)+y^2 B'(T)+x y C'(T)+ x D'(T) + y E'(T) + F'(T)$.
Matching the appropriate powers of $x,y$ on either side of the equality yields the system of Riccati differential equations \eqref{eq:syste}, with the boundary conditions \eqref{eq:singular} translating into \eqref{eq:cond}.
As previously mentioned, the first order terms $D(T)$ and $E(T)$ in $x$ and $y$ are not required if the inventory risk term is of order $x^2$ or $1$.
\hfill \end{proof}

\bibliographystyle{plain}
\bibliography{mybib}

\end{document}